\documentclass[aps,prx,twocolumn,longbibliography,floatfix,10pt]{revtex4-2}
\usepackage{graphicx}
\usepackage[english]{babel}

\makeatletter
\adddialect\l@en\l@english
\makeatother
\usepackage{amsmath,amssymb,braket}
\usepackage{tikz}
\usepackage[caption=false]{subfig}

\usepackage[percent]{overpic}
\usepackage{booktabs}
\usepackage[colorlinks=true,
            linkcolor=blue,
            citecolor=blue,
            urlcolor=blue,
            hypertexnames=false]{hyperref}
\usepackage{float} 
\makeatletter
\let\newfloat\newfloat@ltx
\makeatother
\usepackage{algorithm}
\newcommand{\cobalt}{DysonNet}
\usepackage[noend]{algpseudocode}
\usepackage{xcolor}

\usepackage{amsthm}
\usepackage{bm}

\theoremstyle{plain}
\newtheorem{theorem}{Theorem}

\newtheorem{proposition}[theorem]{Proposition}

\theoremstyle{definition}
\newtheorem{definition}[theorem]{Definition}

\newtheorem{invariant}{Invariant}
\theoremstyle{remark}
\newtheorem{remark}[theorem]{Remark}

\usepackage{microtype}

\newcommand{\vect}[1]{\mathbf{#1}}

\newcommand{\norm}[1]{\left\lVert #1\right\rVert}

\begin{document}

\title{DysonNet: Constant-Time Local Updates for Neural Quantum States}
\author{Lucas Winter}
\affiliation{Faculty of Physics, University of Vienna, Boltzmanngasse 5, 1090 Vienna, Austria}
\author{Andreas Nunnenkamp}
\affiliation{Faculty of Physics, University of Vienna, Boltzmanngasse 5, 1090 Vienna, Austria}
\date{\today}

\begin{abstract}
Neural quantum states (NQS) provide a flexible variational framework for many-body wavefunctions, but suffer from high computational cost and limited interpretability. We introduce DysonNet, a broad class of NQS that couples strictly local nonlinearities through global linear layers. This structure is analogous to a truncated Dyson series which gives an intuitive interpretation of local wavefunction updates as scattering from static impurities. By resumming the scattering series, single-spin-flip updates can be computed in $\mathcal{O}(1)$ time, independent of system size, using an algorithm we call ABACUS. Implementing DysonNet with the state-space model S4, we obtain up to $230\times$ speedups over Vision-Transformers for computing the local estimator. This corresponds to an asymptotic $\mathcal{O}(N^2)$ improvement in training-time scaling, reaching $\mathcal{O}(N \log^2 N)$ total training complexity in area-law phases. Benchmarks on the 1D long-range Ising model and frustrated $J_1$-$J_2$ chains show that DysonNet matches state-of-the-art NQS accuracy while removing the dominant local-update overhead. More broadly, our results suggest a route to scalable NQS architectures where physical interpretability directly enables computational efficiency.
\end{abstract}

\maketitle

\section{Introduction}

Neural quantum states (NQS) have emerged as a flexible and potentially scalable class of variational wave-function ansätze ~\cite{carleo_solving_2017, lange_architectures_2024}. Recently, neural quantum states have achieved state-of-the-art accuracy across a wide domain of challenging ground state problems like frustrated magnets ~\cite{westerhout_generalization_2020, viteritti_transformer_2023, chen_empowering_2024}, the Hubbard model ~\cite{roth_superconductivity_2025, gu_solving_2025, ibarra-garcia-padilla_autoregressive_2025}, long-ranged quantum systems~\cite{roca-jerat_transformer_2024, trigueros_simplicity_2024}, and fractional quantum Hall problems ~\cite{glasser_neural-network_2018, geier_is_2025}. Many architectures have been proposed including restricted Boltzmann machines (RBM) \cite{carleo_solving_2017, melko_restricted_2019}, convolutional networks (CNN) \cite{liang_solving_2018, fu_lattice_2022, levine_quantum_2019}, recurrent networks (RNN) \cite{hibat-allah_recurrent_2020, doschl_neural_2024}, and Transformers (ViT) \cite{sharir_deep_2020, viteritti_transformer_2023, roca-jerat_transformer_2024, sprague_variational_2024, chen_convolutional_2025}.
One of the main promises of NQS is that their nonlinear structure gives them higher expressive power than conventional Ansätze such as tensor network states (TNS) ~\cite{wu_tensor-network_2023, deng_quantum_2017, paul_bound_2025}. 

The downside of this nonlinear structure is that to update the state after a single spin flip forces a full evaluation of the network.
For restricted Boltzmann machines (RBM), local update schemes exist, but their costs scale with system size ~\cite{carleo_solving_2017, nagy_variational_2019, pan_efficiency_2024, chen_systematic_2022, gao_efficient_2017}.
Autoregressive (AR) factorization of the wavefunction can make sampling more efficient, but it leaves local update costs unchanged ~\cite{sharir_deep_2020, hibat-allah_recurrent_2020}. There are specialized Ansätze like finite receptive field CNNs or local graph neural networks that can trivially be updated in $\mathcal O(1)$ but this is achieved by neglecting any long-range correlations.
To our knowledge, no current method endows deep NQS with constant time local updates while retaining their broad expressivity.

\begin{table}[t]
\centering
\label{tab:complexity_single_col}
\resizebox{\columnwidth}{!}{%
\begin{tabular}{lccc}
\toprule
\textbf{Model} & \textbf{Update} & \textbf{Train} & \textbf{Memory} \\
\midrule
ViT~\cite{vaswani_attention_2017,dosovitskiy_image_2021} & $\mathcal{O}(N^2)$ & $\mathcal{O}(N^3\tau)$ & $\mathcal{O}(N^3)$ \\
ResNet/CNN~\cite{lecun_convolutional_1998,he_deep_2016} & $\mathcal{O}(N^2)$ & $\mathcal{O}(N^3\tau)$ & $\mathcal{O}(N^3)$ \\
RBM~\cite{hinton_training_2002,carleo_solving_2017} & $\mathcal{O}(N)$ & $\mathcal{O}(N^2\tau)$ & $\mathcal{O}(N^2)$ \\
RNN (AR)~\cite{hochreiter_long_1997,hibat-allah_recurrent_2020} & $\mathcal{O}(N)$ & $\mathcal{O}(N^2)$ & $\mathcal{O}(N^2)$ \\
PixelCNN (AR)~\cite{oord_conditional_2016,sharir_deep_2020} & $\mathcal{O}(N)$ & $\mathcal{O}(N^2)$ & $\mathcal{O}(N^2)$ \\
\midrule
\textbf{\cobalt} (Area law) & ${\mathcal{O}(1)}$ & ${\mathcal{O}(N \log^2 N)}$ & ${\mathcal{O}(N )}$ \\ 
\textbf{\cobalt}  & ${\mathcal{O}(1)}$ & $\mathcal{O}(N\mathcal T(N)\tau )$ & ${\mathcal{O}(N )}$ \\ 
\bottomrule
\end{tabular}%
}
\caption{\textbf{Asymptotic complexity of \cobalt\ versus standard NQS architectures with global receptive field. }
\textbf{Update}: Cost of a single spin flip given cached environment tensors. 
\textbf{Train}: Cost per gradient step, including Markov-chain sampling and local Hamiltonian evaluation. Metropolis approaches are sampling-dominated, while autoregressive (AR) approaches are local-estimator-dominated.
$\tau$ represents autocorrelation time; $\mathcal T(N)$, defined in Eq.~\eqref{eq:TN}, is sublinear or logarithmic. 
\textbf{Memory}: Peak memory under batched evaluation strategy. Uniquely among deep NQS, \cobalt\ enables constant-time $\mathcal O(1)$ local updates, allowing for log-linear training costs in area-law regimes. In critical phases, it improves upon prior Markov chain Monte Carlo NQS via subquadratic sweep costs and is the only approach with strictly linear peak memory. Baseline classes are cited row-wise; \cobalt\ scaling follows Theorem~\ref{thm:abacus} and Section~\ref{app:sampling-complexity}.}
\vspace{-0.4cm}
\end{table}

\begin{figure*}
  \centering
  \subfloat[]{%
    \begin{overpic}[width=0.32\textwidth,
      trim=2mm 0mm 0mm 2mm, clip]{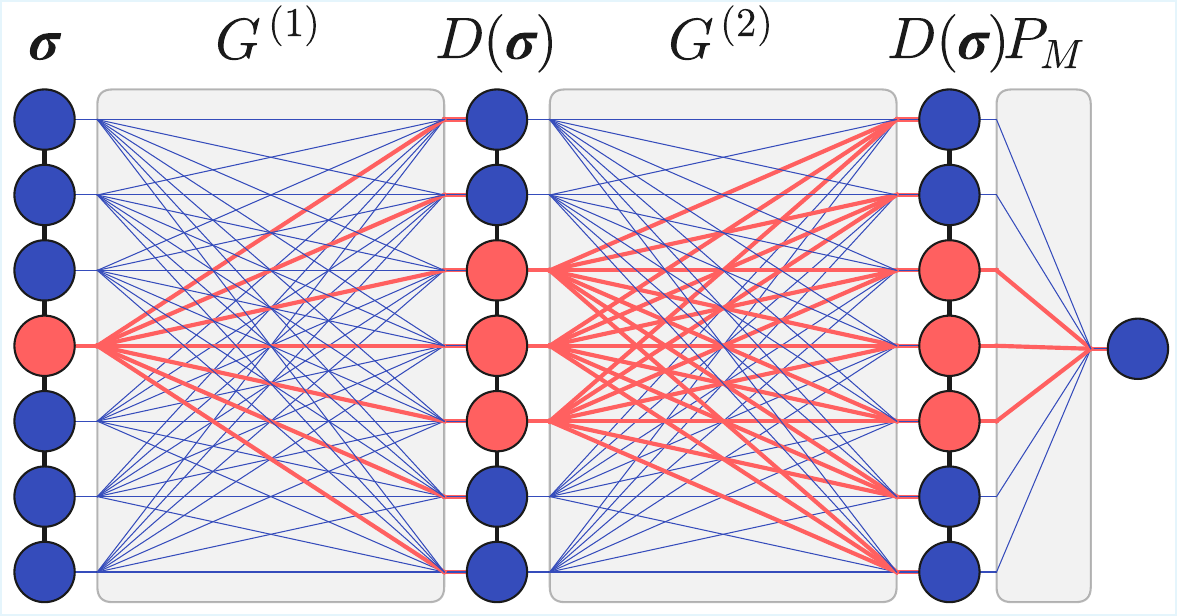}
      \put(-10,45){(a)}%
    \end{overpic}%
    \label{fig:COBALT}}
  \hfill
  \subfloat[]{%
    \begin{overpic}[width=0.32\textwidth,
      trim=15mm 15mm 15mm -1mm, clip]{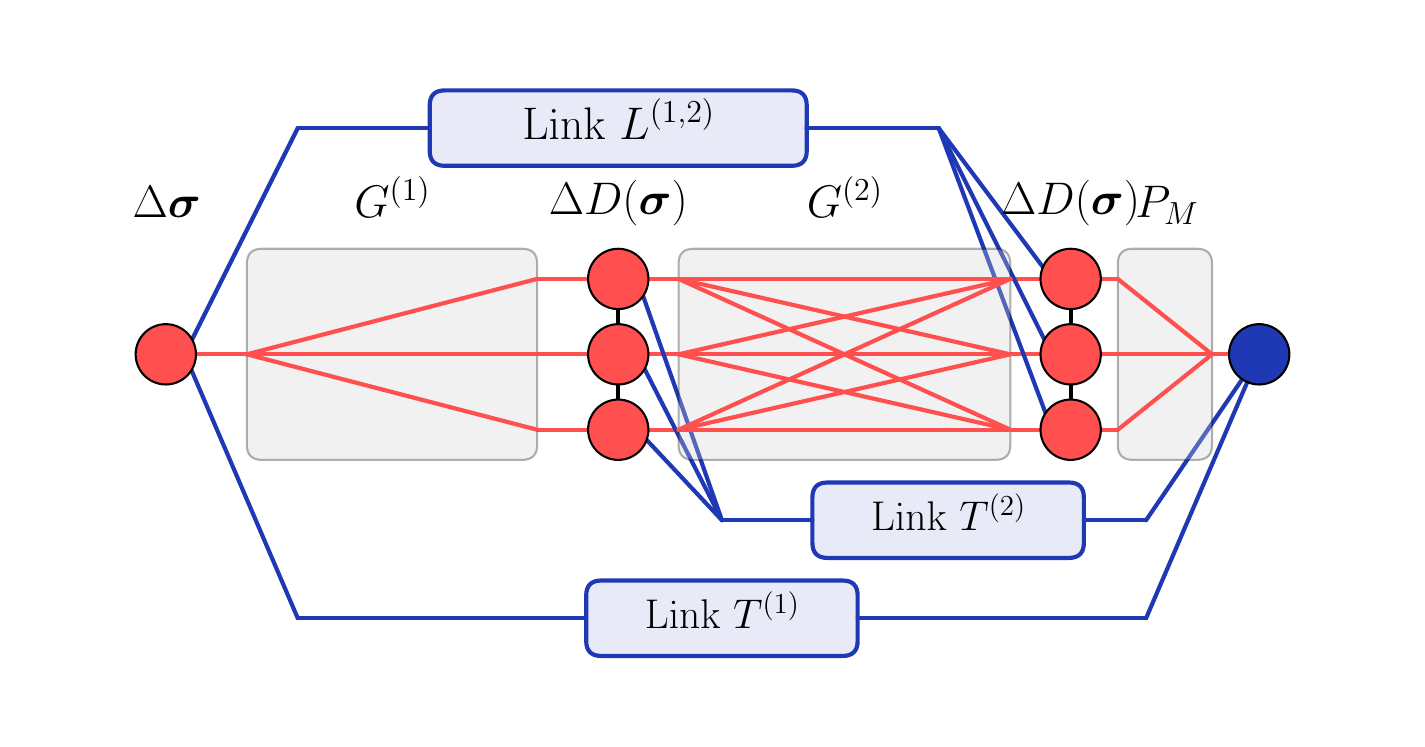}
      \put(0,45){ (b)}%
    \end{overpic}%
    \label{fig:ABACUS}}
  \hfill
  \subfloat[]{%
    \begin{overpic}[width=0.32\textwidth]{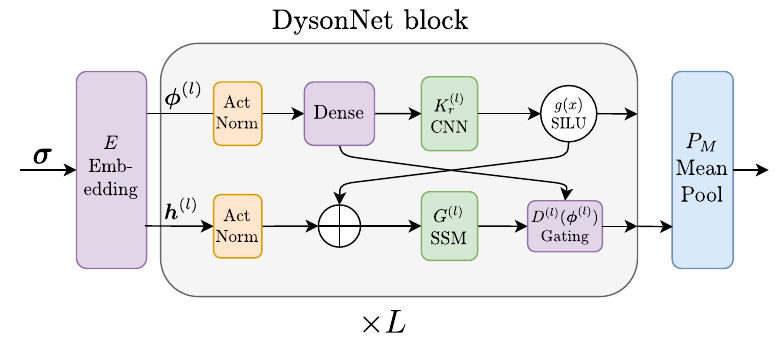}
      \put(-5,45){ (c)}%
    \end{overpic}%
    \label{fig:COBALTPractical}}
  \caption{%
(a) \textit{\cobalt\ Neural network architecture.} Layers alternate global Green’s-function convolutions \(G^{(l)}\) with local nonlinearity \(D^{(l)}(\boldsymbol \sigma)\); the readout is the mean-pooling projector \(P_M\).
A single spin flip \(\sigma_j\!\to\!-\sigma_j\) (red) modifies \(D^{(l)}\) only within a slice of width \(w\) around site \(j\). Subsequent convolutions \(G^{(l)}\) propagate this local change globally building long-range correlations.
(b) \textit{ABACUS local update (Algorithm~\ref{alg:local-update-short}).}
Slice-restricted activations are built recursively using the link tensors \(L^{(l,m)}\) and \(T^{(l)}\); see Eqs.~\eqref{eq:recur_step}–\eqref{eq:recur_out}.
(c) \textit{Practical realization of a \cobalt\ block.}
An embedded spin sequence splits into two streams: \(\boldsymbol\phi\) (CNN path; local, nonlinear) and \(\boldsymbol h\) (SSM path; nonlocal, kept linear in \(\boldsymbol h\)).
Each of the \(L\) stacked blocks updates \(\boldsymbol \phi\) with a dense layer, a small-kernel CNN layer, and a SiLU nonlinearity. In the $\boldsymbol h$ stream we add $\boldsymbol{\phi}$, apply the SSM kernel and apply multiplicative gating via $D(\boldsymbol \phi)=\mathrm{SiLu}(\boldsymbol \phi)$.
After \(L\) blocks, mean pooling over positions yields the readout.
Color/shape key: purple -- dense layers; green -- token mixers (CNN/SSM); orange -- normalization; blue -- projector (mean pooling); circle -- nonlinearity.
}
  \label{fig:triple}
\end{figure*}

This paper closes that gap. We present a constant-time local-update algorithm ABACUS for neural network architectures with linear token mixer and strictly local nonlinearity. The algorithm applies to many popular architectures such as convolutional neural networks~\cite{lecun_convolutional_1998}, linearized attention~\cite{roca-jerat_transformer_2024, rende_are_2025}, Fourier layers~\cite{li_fourier_2021,shah_fourier_2024}, state space models~\cite{gu_efficiently_2022, gu_mamba_2024} or graph neural networks~\cite{scarselli_graph_2009,  kochkov_learning_2021}. We prove that the change in the network output after a single spin flip can be computed in time \emph{independent of $N$} using precontracted \emph{link tensors}. We apply the algorithm to \cobalt, a broad class of convolutional NQS.
Each layer applies a convolution operator and then a local nonlinearity. \cobalt\ admits a particularly simple physical interpretation --- it is functionally analogous to a truncated Dyson series. Within \cobalt, local updates can be interpreted as scattering at static impurities giving ABACUS a physically intuitive interpretation as resumming the scattering series. Furthermore, \cobalt\ allows for efficient calculation of the link tensors via convolutions.

We implement \cobalt\ using the state space model (SSM) S4 ~\cite{gu_efficiently_2022} as a mixer which is particularly interpretable in terms of a physical Green's function. We benchmark our model on long-range spin chains, such as the long-range transverse field Ising model (TFIM). This yields up to two orders of magnitude wall-clock speedups over Vision-Transformers~\cite{roca-jerat_transformer_2024} at $N=1000$ at matched accuracy. We formalize the model class, give the algorithm and proof, describe an efficient implementation, and benchmark accuracy, runtime, and validate scaling to large system sizes.
Our work introduces the widely-applicable algorithm ABACUS and the physically interpretable architecture \cobalt\ that can significantly reduce the computational cost of NQS simulations.

\section{\texorpdfstring{\NoCaseChange{{\normalsize DysonNet}}: Physically motivated neural network architecture}{DysonNet: Physically motivated neural network architecture}}

A key feature of many quantum many-body ground states is a natural separation of scales:
long-wavelength behavior is often universal and can be described by a few collective modes, the basis for effective field theories and the renormalization group.
In contrast, short-range physics can be highly complex and non-universal, depending sensitively on microscopic details of the Hamiltonian. 

We will start by introducing \cobalt, a class of NQS Ansätze that explicitly builds the separation of correlation scales into their architecture. Each block factorizes the representation into (i) a linear, translation-invariant propagator $G$ that captures long-wavelength response and (ii) a local nonlinearity $D(\boldsymbol \sigma)$ that captures short-range structure. Stacking such blocks yields a representation that is both efficient and interpretable, rather than asking a generic ansatz to discover this structure from scratch.

\noindent\paragraph*{\textbf{Definition (\cobalt\  block).}}\label{def:COBALT}
\emph{Let $\boldsymbol \sigma\in \{\pm 1\}^N$ be a spin configuration, $\boldsymbol h^{(l)} \in \mathbb R^{N\times d}$ a sequence of latent vectors with $d$-channels and fix a half–window $w\in\mathbb{N}$. A single \emph{\cobalt}\  block is the map}
\begin{equation}
\;\boldsymbol h^{(l+1)}\;=\;D^{(l)}(\boldsymbol \sigma)\;G^{(l)}\,\boldsymbol h^{(l)}\; .\; \label{eq:COBALT_recurrence}
\end{equation}
\emph{It consists of}

\emph{(1) \textbf{Green's function convolution}, $G$ is a linear, translation–invariant operator along the chain that acts independently on each channel (i.e.~it is diagonal in the channel index $c$).
In real space it acts as a lattice Green’s–function  convolution,}
\begin{align}
(G^{(l)}\boldsymbol h)_{j}\;=\;\sum_{r\in\mathbb{Z}} g_{r,c}^{(l)}\;h_{j-r,c}\,,
\quad c=1,\dots,d.
\end{align}

\emph{(2) \textbf{Local nonlinearity} $D(\boldsymbol \sigma)$ is an input-conditioned matrix multiplication that acts position-wise on each embedding vector in the sequence}
\begin{align}
(D^{(l)}(\boldsymbol \sigma)\boldsymbol h)_{j,c}\;=\;D^{(l)}(\sigma_{j-w:j+w}) \;\boldsymbol h_{j}.
\end{align}
\emph{The receptive field of the local non-linearity is $W = 2w+1$ where $w$ is the half-width.}

\emph{Stacks are obtained by recursive application of \eqref{eq:COBALT_recurrence}.} \\ 

Here, $G$ denotes any translationally invariant kernel convolution, such as Fourier layers~\cite{li_fourier_2021,shah_fourier_2024}, linearized (Toeplitz) attention ~\cite{roca-jerat_transformer_2024}, state space models~\cite{gu_efficiently_2022, gu_mamba_2024}, or analytical Green's functions.
While limited by finite receptive fields, standard CNN kernels also apply~\cite{lecun_convolutional_1998}.
Depending on the padding scheme, the convolution can support either open or periodic boundary conditions.

The network is composed of a stack of \cobalt\  blocks followed by a pooling operation, (multiple) densely connected layers, and a nonlinearity.
This means the network output is 
\begin{align}
    \Psi(\boldsymbol \sigma) =f\bigg[ A  P_M\bigg(G^f\sum_{M=0}^L\prod_{l=1}^M D^{(l)} (\boldsymbol \sigma) G^{(l)} \bigg) E \boldsymbol \sigma \bigg] , \label{eq:wavefunc}
 \end{align}
\noindent where $E\in \mathbb R^{d\times 1}$ is the per-site embedding layer (applied independently at each position), $(P_M\boldsymbol h)_c = \frac 1 N \sum_j h_{j,c}$ is a mean pooling operation. The sum adds a residual connection from the latent activations of each layer to the output. $G^f$ is the final Green's function convolution. $f$ is the final nonlinearity often chosen to be $f(x)= \log\cosh(x)$. $A$ is the final matrix multiplication chosen to be real valued $A \in \mathbb R^{d\times d_{\mathrm{out}}}$ for stoquastic Hamiltonians or $A \in \mathbb C^{d\times d_{\mathrm{out}}}$ for complex wavefunctions. 

While the network definition \eqref{eq:wavefunc} seems quite complicated at first glance it admits a very simple interpretation.
For simplicity, let us assume weight sharing between the different layers $G^{(l)} = G_f := G_0$ and $D^{(l)}(\boldsymbol \sigma) = D(\boldsymbol \sigma)$~\cite{yamazaki_physics-inspired_2026}. We can then interpret the product of layers in Eq. \eqref{eq:wavefunc} (expression inside round bracket) as an effective propagator $G$.
Expanding this term yields
\begin{align}
\begin{split}
    G &= G_0\sum_{M=0}^L \prod_{l=1}^M  D(\boldsymbol \sigma)G_0 \\ &=G_0 + G_0 D(\boldsymbol \sigma) G_0 + G_0 D(\boldsymbol \sigma) G_0 D(\boldsymbol \sigma) G_0+\dots. \label{eq:Dyson}
\end{split}
\end{align}
We see the network can be interpreted as a truncated Dyson series. Specifically, we can identify $G_0$ with the free propagator, and $D(\boldsymbol \sigma)$ with the coupling vertex. This further justifies our decision to restrict $D(\boldsymbol \sigma)$ to be local. While this correspondence allows one to initialize the network using analytic theory, we learn the weight matrices via gradient descent.

The functional equivalence of the network to a Dyson series clarifies how its parameters control the correlations it can express. With $L$ layers, the network can represent any non-local power-law or exponentially decaying interaction up to coupling order $L$; long-ranged correlations are therefore truncated in coupling order (see appendix \ref{app:power-law}). We stress that by parameterizing $G$ as a power law, the network can naturally express correlations at criticality making it more expressive than a fixed bond dimension matrix product state. By contrast, the local nonlinearity can, in principle, generate correlations of arbitrarily high coupling order, but only within a finite width $W$. Varying the number of layers $L$ and the nonlinear receptive field $W$ thus provides a systematic way to improve the expressivity of the model.  Nevertheless, there exist ground states with genuinely non-local correlations of arbitrary coupling range, such as phases with intrinsic topological order characterized by non-local Wilson loop operators in 2D ~\cite{chen_local_2010,wen_colloquium_2017, Levin_String-net_2005}, for which these truncations are insufficient. Representing such exotic quantum states requires additional techniques.

Model training is performed via Variational Monte-Carlo and Stochastic Reconfiguration ~\cite{carleo_solving_2017}. We minimize the variational energy by sampling configurations $\boldsymbol \sigma$ and computing the local energy
\begin{equation}
E_{\mathrm{loc}}(\boldsymbol \sigma)
= \sum_{\boldsymbol \sigma'}
\langle \boldsymbol \sigma \vert H \vert \boldsymbol \sigma' \rangle
\frac{\Psi(\boldsymbol \sigma')}{\Psi(\boldsymbol \sigma)}. \label{eq:local-estimator}
\end{equation}
Since both the Metropolis sampling and the local energy calculation involve configurations $\boldsymbol \sigma'$ that differ from $\boldsymbol \sigma$ by only a single spin flip, recomputing the full wavefunction is inefficient. In the next section, we derive an optimized update algorithm to handle these local changes exploiting that the network is equivalent to a Dyson series.

\begin{algorithm}[H]
  \caption{\textsc{ABACUS} local-update kernel}
  \label{alg:local-update-short}
  \begin{algorithmic}[1]
    \Require Link tensors $\{T^{(l)}\}_{l=1}^L$, $\{L^{(l,m)}\}_{1 \le m \le l-2,\ 2\le l\le L}$
    \Ensure Mean-pooled vector $\boldsymbol\Omega$
    \State $\tilde{\boldsymbol h}^{(1)} \gets \Delta D^{(1)}\bigl(M^{(1)}\,\Delta\boldsymbol \sigma +\tilde{\boldsymbol h}^{(1)}_0\bigr)$
    \For{$l = 2,\dots,L$}
      \State $\tilde{\boldsymbol h}^{(l)} \gets \Delta D^{(l)}\big[ \tilde M^{(l)} \tilde{\boldsymbol h}^{(l-1)}
        + \sum_{m=1}^{l-2} L^{(l,m)} \tilde{\boldsymbol h}^{(m)}
        + \tilde{\boldsymbol h}^{(l)}_0\bigg]$
    \EndFor
    \State $\boldsymbol \Omega \gets A\!\left(P_M \tilde{\boldsymbol h}^{(L)}
      + \sum_{l=1}^{L} T^{(l)} \tilde{\boldsymbol h}^{(l)}\right)$
    \State \Return $\boldsymbol \Omega$
  \end{algorithmic}
\end{algorithm}

\section{ABACUS: Asymptotically optimal local updates}

We will now derive the update algorithm for the amplitude $\Psi(\boldsymbol \sigma)$ after a single spin flip at position $j$ relative to $\Psi(\boldsymbol \sigma_0)$. Importantly, the results of this section apply to any network with linear token mixers $M^{(l)}$, not only convolutional token mixers such as in \cobalt. For the \cobalt\ specialization, we identify $M^{(l)}=G^{(l)}$ and $M^f=G^f$. As the local vertex $D(\boldsymbol \sigma)$ only changes in a narrow slice of width $W=2w+1$ around the flip $j$ (with half-width $w$; see Fig.~\ref{fig:COBALT}), we can interpret the update as a scattering process off a localized impurity (the perturbation). This perspective allows us to derive an efficient update algorithm by exactly resumming the Dyson series defined by the network.

We formalize this by decomposing every layer’s nonlinearity $D^{(l)}(\boldsymbol\sigma)$ into a static background  $D_0^{(l)}:=D^{(l)}(\boldsymbol\sigma_0)$ and a localized perturbation, $\Delta D^{(l)} = D^{(l)}(\boldsymbol \sigma) - D^{(l)}(\boldsymbol \sigma_0)$ that is non-zero only on a narrow slice around $j$.  Substituting this decomposition into Eq.~\eqref{eq:wavefunc} yields 
\begin{equation}
\boldsymbol\Omega(\boldsymbol\sigma)
=
AP_M\!\left(M^f\sum_{M=0}^L\prod_{l=1}^{M} (D^{(l)}_0 + \Delta D^{(l)})\,M^{(l)}\right)E\boldsymbol\sigma,
\label{eq:fullprod}
\end{equation}
\noindent where $\boldsymbol\Omega(\boldsymbol \sigma)$ is the activation before the final nonlinearity $\Psi(\boldsymbol \sigma) = f(\boldsymbol \Omega(\boldsymbol \sigma))$. Intuitively, $D_0$ acts as the background medium and $\Delta D$ as a local scatterer; $M^{(l)}$ is the propagator spreading the local perturbations and the projector $P_M$ performs the measurement.

 We now need to resum the product in \eqref{eq:fullprod}. If we naively expand the ordered product in Eq.~\eqref{eq:fullprod} we formally obtain exponentially many \emph{matrix words} $\mathcal O(2^L)$ in the network depth $L$. This would be prohibitive for deeper networks. To circumvent this we propose ABACUS (Algorithm~\ref{alg:local-update-short}) for computing matrix-vector products like \eqref{eq:fullprod} in quadratic time $\mathcal O(L^2)$ in network depth. The algorithm works by
 (i) \emph{pre-contracting} all background propagators $D_0^{(\cdot)}M^{(\cdot)}$ into \emph{environment (link) tensors}, and (ii) \emph{dynamically resumming} the matrix words with a slice-local recurrence. 

The link tensors defined for a spin configuration $\boldsymbol \sigma_0$ propagate activations through the frozen background 
\begin{align}
T^{(l)} &= P_MM^f\left(\sum_{M=l}^L\ \prod_{n=l+1}^{M} D_0^{(n)}M^{(n)}\right)P_{S,j}^{\top},
\label{eq:Tl} \\
L^{(l,m)} &= P_{S,j}M^{(l)}\left(\prod_{n=m+1}^{l-1} D_0^{(n)}M^{(n)}\right)P_{S,j}^{\top},
 \ , 
\label{eq:Lmn}
\end{align}

\noindent with $1\le m\le l-2$. Note, the projector $P_{S,j}$ restricts the activation on a narrow slice $[j-w,\dots j+w]$ around the update $j$. The link tensor $T^{(l)}:\mathbb R^{W\times d}\to \mathbb R^d$  \eqref{eq:Tl} lifts slice-restricted activations back to full size and propagates them to the output.  $L^{(l,m)} : \mathbb R^{W\times d}\to \mathbb R^{W\times d}$ \eqref{eq:Lmn} propagates slice-restricted activations through the background between layers $m$ and $l$. Crucially, the link tensors require only linear memory in $N$. For the special case of \cobalt\ they can be constructed in log-linear time $\mathcal O(N\log N)$ (see appendix \ref{app:linkTensorConstruction}) for both periodic or open boundary conditions. Physically, the link tensors represent fully dressed Green's function propagators relative to some background. 

Using the link tensors \eqref{eq:Tl} and \eqref{eq:Lmn}, an update to the wave function amplitude $\boldsymbol \Omega(\boldsymbol \sigma)$ can be computed by summing the scattering events. Let us define the slice-restricted activations $\boldsymbol {\tilde h}^{(l)} = P_{S,j} \boldsymbol h^{(l)} \in \mathbb R^{W\times d}$ and $\boldsymbol {\tilde h}_0^{(l)} = P_{S,j} \boldsymbol h_0^{(l)} \in \mathbb R^{W\times d}$. Furthermore, we define the slice-restricted convolution $\tilde M^{(l)} = P_{S,j} M^{(l)} P_{S,j}^{\top} \in \mathbb R^{W\times W}$. The slice-restricted activations can be computed using the following recurrence relation 

\begin{align}
\tilde{\boldsymbol h}^{(l)} 
&= 
\Delta D^{(l)}\left(\tilde{\boldsymbol h}^{(l)}_0+
\tilde M^{(l)}\tilde{\boldsymbol h}^{(l-1)}
+ \sum_{m=1}^{l-2} L^{(l,m)}\tilde{\boldsymbol h}^{(m)}
\right),
\label{eq:recur_step} \\ 
\boldsymbol\Omega
&=
A\!\left(
P_M\,\tilde{\boldsymbol h}^{(L)}
+
\sum_{l=1}^{L} T^{(l)}\,\tilde{\boldsymbol h}^{(l)}
\right).
\label{eq:recur_out}
\end{align}

\noindent This recurrence defines the ABACUS algorithm. Physically, $\tilde{\boldsymbol h}^{(l)}$ represents the sum of all paths scattered at layer $l$. This amplitude sums three components: the unscattered incident wave $\boldsymbol{\tilde h}_0^{(l)}$; all paths $\boldsymbol {\tilde h}^{(l-1)}$ scattered at the previous layer $l-1$; and the cumulative contributions from all prior defects $m$, propagated through the background.

In the appendix \ref{sec:induction_proof} we prove the following statement 

\begin{theorem}[ABACUS exactness and complexity]
\label{thm:abacus}
Given a linear token mixer $M^{(l)}$ and a local non-linearity $D^{(l)}(\boldsymbol \sigma)$, the recurrence
\eqref{eq:recur_step}–\eqref{eq:recur_out} evaluates
$\boldsymbol \Omega(\boldsymbol \sigma)$ exactly for any local update $\boldsymbol \sigma=\boldsymbol \sigma_0+\Delta\boldsymbol \sigma$.
Given precomputed links of Eq.~\eqref{eq:Tl} and \eqref{eq:Lmn}, the runtime $\mathcal{O}(L^2W^2 d^2)$ is constant in $N$
 and the memory  
$\mathcal{O}(NL^2W d^2)$ is linear in $N$.
\end{theorem}

We note ABACUS applies to any network architecture with a linear token mixer and local nonlinearity, including linearized self-attention and graph neural networks, and is not limited to convolutions. In Appendix~\ref{app:ABACUS2D} we describe ABACUS in 2D beyond the case considered here. The runtime complexity remains unchanged and depends on the number of tokens $W$ conditioning the nonlinearity. Note, however, that the structure of $M$ dictates the link-tensor precomputation cost (Eqs.~\eqref{eq:Tl} and \eqref{eq:Lmn}): while $\mathcal O(N \log N)$ for the \cobalt\ class $M^{(l)} = G^{(l)}$, it can increase to $\mathcal O(N^2)$ for dense, unstructured mixers. We discuss this in detail in the runtime section.

The constant-time complexity $O(1)$ in Theorem~\ref{thm:abacus} strictly holds when the hyperparameters $w$ and $L$ are independent of system size $N$. Physically, this assumption relies on a separation of scales: complex, non-universal correlations are confined to a short range $w$, while long-wavelength behaviors are captured by the propagator $M^{(l)}$ ($G^{(l)}$).

A common intuition from perturbation theory is that reproducing algebraic decay (criticality) from a short-range bare propagator requires resumming an infinite Dyson series ($L \to \infty$). However, in \cobalt, the propagators $G^{(l)}$ are learnable, translation-invariant kernels capable of representing power-law decays directly. Consequently, $G^{(l)}$ should be interpreted not as a bare propagator, but as an effective dressed propagator. We empirically observe that even at criticality, small $L$ is sufficient to preserve the $O(1)$ update cost.

The algorithm extends naturally to spin clusters of size $n$ by increasing the effective slice width to $w_{\mathrm{cluster}} = w + \lceil n/2\rceil$. Regarding complexity, while the update step is size-independent, the initial link-tensor construction incurs a quasi-linear $\mathcal{O}(N\log N)$ cost. For estimators requiring $\mathcal{O}(N)$ evaluations, this overhead amortizes to $\mathcal{O}(\log N)$ per sample. In the next section, we outline a strategy to amortize the link tensor costs for the sampler.

\section{Screened typewriter sampler}
\label{sec:sampler}

\begin{algorithm}[H]
  \caption{\textsc{ScreenedTypewriterSweep}($\boldsymbol \sigma, s, \epsilon$)}
  \label{alg:screened-typewriter}
  \begin{algorithmic}[1]
    \Require Configuration $\boldsymbol \sigma_0$, spacing $s$, tolerance $\epsilon$
    \Ensure Updated configuration $\boldsymbol \sigma$
    \Comment{Latents $\boldsymbol{\Omega}$ \& Tensors $\mathcal{T}$}
    \State Select spin flips $\Delta \sigma_1, \dots \Delta \sigma_K$ with spacing $s$
    \State Sample random numbers $u_{1\dots K} \sim \mathcal{U}[0, 1]$
    \State $\boldsymbol{\Omega}, \mathcal{T} \gets \textsc{LinkTensor}(\boldsymbol \sigma)$ 
    \State $\boldsymbol \sigma \gets \boldsymbol \sigma_0$
    \State $\{\Delta \boldsymbol{\Omega}_k\} \gets \mathrm{ABACUS}(\{\Delta \sigma_k\}, \boldsymbol \sigma_0, \mathcal{T})$ \Comment{Get update vectors}
    \For{$k=1 \dots K$}
      \State $\tilde{R} \gets |f(\boldsymbol{\Omega} + \Delta \boldsymbol{\Omega}_k) / f(\boldsymbol{\Omega})|^2$ \Comment{Fast activation ratio}
      \State $P_{\text{acc}} \gets \min(1, \tilde{R})$
      
      \If{$|u_k - P_{\text{acc}}| \le \epsilon$} \Comment{\textbf{Ambiguous:} Refresh}
        \State $\boldsymbol{\Omega}, \mathcal{T} \gets \textsc{LinkTensor}(\boldsymbol \sigma)$ 
        \State $\{\Delta \boldsymbol{\Omega}_k\} \gets \mathrm{ABACUS}(\{\Delta \sigma_k\}, \boldsymbol \sigma_0, \mathcal{T})$ 
        \State $\tilde{R} \gets |f(\boldsymbol{\Omega} + \Delta \boldsymbol{\Omega}_k) / f(\boldsymbol{\Omega})|^2$ 
        \State $P_{\text{acc}} \gets \min(1, \tilde{R})$
      \EndIf

        \If{$u_k < P_{\text{acc}}$} 
            \State $\boldsymbol \sigma \gets \boldsymbol \sigma+ \Delta \sigma_k$
            \State $\boldsymbol{\Omega} \gets \boldsymbol{\Omega} + \Delta \boldsymbol{\Omega}_k$ \Comment{Update latents}
        \EndIf

    \EndFor
    \State \Return $\boldsymbol \sigma$
  \end{algorithmic}
\end{algorithm}

We propose a Metropolis sampler that addresses two bottlenecks: standard Metropolis updates are inherently sequential and thus underuse GPUs, and our link-tensor construction must be amortized to achieve sublinear cost in system size. The key idea is to treat a batch of spatially separated spin flips as a dilute gas of defects. When the defects are well separated, their scattering signatures do not interfere. We can then evaluate the flips as independent proposals in parallel and reuse the same background link tensors across the batch solving both problems. This approximation introduces a small error in the acceptance ratio, but a screened acceptance rule restores exact detailed balance \cite{semon_lazy_2014}.

Proposals are single-site flips, with occasional global spin inversions to aid mixing. A proposal $\boldsymbol \sigma_0\to\boldsymbol \sigma$ is accepted with $P_{\rm acc}=\min\{1,R\}$, where $R=|\Psi(\boldsymbol \sigma)|^2/|\Psi(\boldsymbol \sigma_0)|^2$. To amortize the $\mathcal{O}(N\log N)$ cost of constructing the background link tensors, we evaluate a batch of $F$ single-site proposals in parallel using the static background of the reference $\boldsymbol \sigma_0$. This relies on the \emph{Independent Scattering Approximation}: we assume the defects interact with the background but not with each other. Note we are using the term in analogy to scattering theory in disordered media \cite{norris_multiple_2011, mishchenko_independent_2018}, although here the defects are not physical. Detailed sampler settings are listed in Appendix Sec.~\ref{app:training_details}.  
The error induced by neglecting the interaction (multiple scattering) between simultaneous flips is
\begin{align}
\epsilon(\boldsymbol \sigma)
:= \left|\frac{\Psi_{\mathrm{ISA}}(\boldsymbol \sigma)-\Psi(\boldsymbol \sigma)}{\Psi(\boldsymbol \sigma_0)}\right|.
\end{align}
This error is controlled by the decay of the Green's function, see appendix \ref{app:metropolis-error}. By enforcing a minimum separation $s$ between defects, we ensure the propagator connecting any two flips is negligible. 

To enforce a minimum spacing between flips, we use a typewriter schedule \cite{bohm_understanding_2018,hasenbusch_anisotropic_2011} (see algorithm \ref{alg:screened-typewriter}). We partition the lattice into blocks of width $s\geq w$. We sweep first over even blocks, then odd blocks. In each active block, we select one site uniformly. This creates a superlattice of defects separated by the spacing $s$. We build the link tensors once per sweep (costing $\mathcal{O}(N\log N)$) and reuse them for all proposals in the batch. Because the proposals are spatially separated, the independent scattering approximation is highly accurate and the error $\boldsymbol \epsilon$ can be bounded (see appendix \ref{app:metropolis-error}).

Because \(\epsilon\) is bounded, the Metropolis acceptance rule can be enforced exactly even though the wave-function amplitude is approximate. We recover exact detailed balance using a screened acceptance rule \cite{semon_lazy_2014}. We accept a proposal based on the approximate ratio $\tilde R = |\Psi_{\mathrm{ISA}}(\boldsymbol \sigma)|^2 /|\Psi(\boldsymbol \sigma_0)|^2$ and error bound $\epsilon$.
For a random number $u \sim \mathcal{U}[0,1]$
\begin{enumerate}
    \item Fast Accept: If $u < \tilde R - \epsilon$, accept immediately.

    \item Fast Reject:  If $u > \tilde R + \epsilon$, reject immediately.

    \item Correction: Only in the rare window $u \in [\tilde R - \epsilon, \tilde R + \epsilon]$ do we evaluate the exact interference terms (recomputing the full amplitude).
\end{enumerate}
\noindent This scheme allows parallel evaluation without statistical bias; proof see appendix \ref{app:screenedProof}.

The efficiency depends on the physical correlation length of the system.
The decay of correlations determines the required spacing $s$ to keep $\epsilon$ small.
In appendix \ref{app:sampling-complexity}, we derive the amortized per proposal cost
\begin{align}
\mathcal T(N)=
\begin{cases}
\mathcal{O}(\log^2 N) & \text{Area-law (Gapped)},\\
\mathcal{O}(N^{\gamma}\log N) & \text{Power-law (Critical)},\\
\mathcal{O}(N\log N) & \text{All-to-all}.
\end{cases}. \label{eq:TN}
\end{align}
For gapped phases (area-law), the Green's function decays exponentially, allowing dense packing of updates and polylogarithmic runtime. For critical systems (power-law), the spacing $s$ must grow with system size, leading to sublinear scaling $N^\gamma$ (where $\gamma < 1$ depends on the critical exponent $\alpha$). Only in the all-to-all limit does the speedup vanish. In our numerics, we find the empirical error $\epsilon(\boldsymbol \sigma, \boldsymbol \sigma_0)$ to be small up to very large system sizes $N\leq 1000$ and even on criticality. For practical system sizes, adjusting the spacing $s$ is unnecessary. This means the method yields a speedups of $\mathcal{O}(N)$ significantly exceeding the worst case analytical estimates (see appendix \ref{app:throughput}). Another significant practical advantage is that the screened acceptance rule enables reduced precision sampling. While neural quantum states typically require high floating-point precision, this rule permits lower-precision formats, which can yield additional constant speedups of $2\times$ to $4\times$.

Table \ref{tab:complexity_single_col} contrasts the runtime complexity of \cobalt\  against standard NQS architectures. Among deep networks, our approach uniquely enables constant-time, $\mathcal O(1)$ spin-flip updates given cached environment (link) tensors. In comparison, Restricted Boltzmann Machines (RBM) scale linearly as $\mathcal O(N)$. Autoregressive (AR) models also face costs of $\mathcal O(N)$ per local update. While AR models can be sampled efficiently, the flip costs induce $\mathcal O(N^2)$ work for local observables which is their most significant bottleneck. While matrix product states (MPS) and local CNNs also achieve $\mathcal O(1)$ updates, they occupy a restricted expressivity class and fail to efficiently represent long-range power-law correlations.

The fact that local updates can be evaluated in constant time gives a distinctive advantage to \cobalt\ +ABACUS when evaluating the Hamiltonian \eqref{eq:local-estimator}. Local Hamiltonians generally have $\mathcal O(N)$ connected matrix elements, requiring an equivalent number of network evaluations. RBMs, RNNs and PixelCNNs face a cost of $\mathcal O(N)$ for local updates resulting in a cost of $\mathcal O(N^2)$ in total. The vision transformer (ViT) network evaluation scales as $\mathcal O(N^2)$ leading to a cubic cost $\mathcal O(N^3)$.  \cobalt\  improves on all prior approaches by constructing link tensors in $\mathcal O(N\log N)$ and using the ABACUS algorithm for updates in $\mathcal O(1)$ given cached link tensors, achieving a log-linear cost $\mathcal O(N \log N)$.

Metropolis-Hastings approaches generally suffer from sampling costs of $\mathcal{O}(N\tau C_{\text{network}})$, leading to scaling of $\mathcal{O}(N^3 \tau)$ for ViTs and $\mathcal{O}(N^2 \tau)$ for RBMs. \cobalt+ABACUS improves upon these by ensuring the per-sweep cost is strictly sub-quadratic, $\mathcal{O}(\mathcal T(N) N \tau)$, and even log-linear, $\mathcal{O}(N \log^2 N)$, in area-law phases. Autoregressive architectures, however, retain the fastest asymptotic sampling capability with a cost of $\mathcal{O}(N)$.

When considering total training time, Metropolis approaches remain dominated by the sampling cost, meaning our algorithmic improvements yield asymptotic speedups independent of ground state correlations. Autoregressive models, by contrast, are computationally dominated by the $\mathcal{O}(N^2)$ overhead required to evaluate the local Hamiltonian. As a result, in area-law phases where $\tau \propto \mathcal{O}(1)$, \cobalt+ABACUS achieves a total training complexity of $\mathcal{O}(N \log^2 N)$, which is lower than that of autoregressive approaches. Nevertheless, for systems with divergent autocorrelation time autoregressive approaches remain asymptotically faster. While slower, using Markov chain sampling has the advantage of  efficient enforcement of wavefunction symmetries that are difficult to implement in  autoregressive architectures. 

We noted in the algorithm section that ABACUS applies to any NQS with global linear layers. For dense, unstructured linear mixers, a forward pass typically costs $\mathcal{O}(N^2)$, and constructing/caching the link tensors has the same scaling, $\mathcal{O}(N^2)$, rather than $\mathcal{O}(N\log N)$ for \cobalt. In this regime, ABACUS reduces the cost of a Metropolis sweep of local estimators (i.e., $N$ single-site proposals) from $\mathcal{O}(N C_{\mathrm{fwd}})=\mathcal{O}(N^3)$ to $\mathcal{O}(C_{\mathrm{fwd}}+N)=\mathcal{O}(N^2)$ (for fixed depth and locality window). The corresponding sampling speedups depend on the applicability of the independent-scattering approximation; in this work we only provide rigorous bounds for \cobalt, and we leave a systematic study of sampling speedups for dense mixers to future work. 

Finally, we compare asymptotic peak memory usage. For standard autoregressive and Markov chain Monte Carlo architectures, the local estimator introduces a significant bottleneck; evaluating the wavefunction on $\mathcal{O}(N)$ connected configurations typically incurs an $\mathcal{O}(N^2)$ memory cost. \cobalt\ avoids this quadratic scaling. Because local updates require only slice-local activations and the link tensors scale linearly, the total memory footprint is reduced to $\mathcal{O}(N)$. This scaling behavior allows for the simulation of large systems on commodity GPUs that lack the high memory capacity of data-center hardware.

Beyond large system asymptotics, network hyperparameters determine when ABACUS outperforms \cobalt. We illustrate this using the local estimator. With $\mathcal O(N)$ forward passes, the cost of computing the local estimator under \cobalt\  is $\mathcal O( NL (d N\log N +d^2 N))$. For ABACUS, the dominant cost is computing the $(L-1)^2$ interlayer link tensor, which scales as $\mathcal O(N\log N d^2 W)$. Applying the ABACUS recurrence costs $\mathcal O(L d^2 W^2)$. Thus, ABACUS is generally more efficient when.

 \begin{align}
     N \gtrsim C_{\text{cache}}  W d (L-1) + C_{\text{ABACUS}} W^2.
 \end{align}

 \noindent where $C$ is an implementation-dependent constant; for our code, $C\approx 1/6$ (see Fig. \ref{fig:kernel_runtime}). Efficiency depends quadratically on the non-linearity’s receptive field $W$, and linearly on both the embedding dimension $d$ and layer count $L$. Patching can minimize $W$ by allowing multiple spins to contribute to a single embedding vector. Although this increases the embedding vector size, the trade-off is favorable. While very deep networks are inefficient with ABACUS, they are unnecessary with \cobalt\ because FFT-based convolutions provide a global receptive field. Very wide networks also reduce efficiency. For these, we recommend giving $D(\boldsymbol \sigma)$ a block-diagonal structure in channel space creating a mixture-of-experts like structure. Dividing into blocks of fixed size $\Delta d$ implies efficiency when $N  \gtrsim C W^2 \Delta d (L-1)$ making ABACUS scalable to arbitrary wide networks.

\section{\texorpdfstring{Implementation of \NoCaseChange{{\normalsize DysonNet}}}{Implementation of DysonNet}}

So far we have treated the free propagator $G_0$ and the coupling vertex $D(\boldsymbol \sigma)$ as abstract objects, since our algorithms only require their general properties (linearity for $G_0$, locality for $D$). The reader is invited to instantiate $G_0$ and $D$ with their favorite sequence model, such as CNNs ~\cite{lecun_convolutional_1998}, Fourier layers ~\cite{li_fourier_2021, shah_fourier_2024}, or certain forms of linearized self-attention ~\cite{roca-jerat_transformer_2024}. Here, we will parametrize $G_0$ as a state space model and $D(\boldsymbol \sigma)$ as a short-range CNN (see Fig \ref{fig:COBALTPractical}).

\emph{Parameterizing the Green's function $G^{(l)}$.}
State space models are a natural choice as they allow to interpret $G^{(l)}$ as a free propagator. This is consistent with viewing the network as 
 a truncated Dyson series. A state space model is a sequence-to-sequence model where the output is obtained by computing the response of a linear state-evolution equation to the input. We implement a bidirectional state-space model mapping a sequence of $\boldsymbol h_j \to \boldsymbol {\tilde h}_j$ by evolving the state vector $\boldsymbol x_j\in\mathbb C^{s\times d}$ via the following discretized differential equation
\begin{align}
x_{j,c} \;&=\; A_c\,(x_{j+1,c}+x_{j-1,c} - x_{j,c} ) + B_c\,h_{j,c},\qquad \\ \tilde h_{j,c} \;&=\; C_c\,x_j , \label{eq:ssm_discretized}
\end{align}
where  $A_c\in \mathbb C^{s\times s},  B_c \in \mathbb C^{1\times s}, C_c\in \mathbb C^{s\times 1}$ are the learned state matrices. The first term in \eqref{eq:ssm_discretized} is the lattice Laplace operator spreading an excitation. As the system is translationally invariant, the response can be calculated using the Fast Fourier transform $\boldsymbol {\tilde h} = \mathcal F^{-1}\!\bigl(\, G(k)\;\odot\;\mathcal F(\boldsymbol h)\,\bigr)$ and the Green's function 
\begin{align}
G_c(k) =  C_c\,(I - e^{i2\pi k/N}A_c)^{-1} B_c \; \label{eq:resolvent}  .
\end{align}
\noindent When $A_c=\mathrm{diag}(\lambda_1,\dots,\lambda_s)$ is diagonal, the resolvent and
$G(k)$ reduce to a sum of Cauchy kernels (simple poles)

\begin{align}
 \; G_c(k) \;=\; \sum_{j=1}^{s} \frac{C_{c,j} B_{c,j}}{\,1 - e^{i2\pi k/N}\lambda_{c,j}\,}.
\end{align}
 Physically, these are damped/oscillatory
modes with poles $\lambda_{c,j}$ inside the unit disk ($|\lambda_{c,j}|\leq1$) on the complex plane. The Green's function corresponds to the free propagation of an excitation (for example domain wall or magnon mode) in space. Note as the number of modes grows, $G_c(k)$ can also approximate power laws over finite range. 

In our paper we use the S4 Ansatz for $A= \Lambda + \boldsymbol p \boldsymbol q^T$  consisting of a diagonal matrix $\Lambda \in \mathbb C^{s\times s}$ and a rank one perturbation $\boldsymbol p, \boldsymbol{q} \in \mathbb C^{1\times s}$~\cite{gu_efficiently_2022}. We provide the analytical Green's function $G(k)$ in appendix \ref{app:cobalt-details}.  Making this choice is physically analogous to considering the interacting Green's function of an excitation scattered at a static impurity.

\emph{Parameterizing the local nonlinearity $D(\boldsymbol \sigma)$.} For parameterizing the local vertex $D(\boldsymbol \sigma)$ we use a short-ranged CNN kernel, as this naturally enforces locality via a finite receptive field. From the input configuration we construct a second data stream $\boldsymbol{\phi}^{(l)}$  (see Fig. \ref{fig:COBALTPractical}). Each layer applies (i) a site-wise dense map, (ii) a short-range convolution, and (iii) a nonlinearity. The dense map $D_j(\boldsymbol \sigma)$ is then conditioned on the data stream $\boldsymbol{\phi}^{(l)}$
\begin{align}
    D_j(\boldsymbol \sigma)  &= U^{(l)} \,\mathrm{diag}\!\big(\chi(\boldsymbol \phi_j^{(l-1)})\big)\, V^{(l)} . \label{eq:Dpractical}
\end{align}
Here, $\chi$ is a nonlinearity that acts locally on each vector $\boldsymbol \phi_j^{(l-1)}$. The matrices  $U^{(l)},V^{(l)}\!\in\!\mathbb{R}^{d\times d}$ are layer-specific and configuration-independent. Thus $D_j(\boldsymbol \sigma)\!\in\!\mathbb{R}^{d\times d}$ has fixed eigenvectors but spin-configuration–dependent eigenvalues through $\boldsymbol\phi^{(l-1)}_j$. Besides the multiplicative gating $D(\boldsymbol \sigma)$ \eqref{eq:Dpractical} we also inject the $\boldsymbol{\phi}_j^{(l)}$ stream into the $\boldsymbol \sigma^{(l)}\to \boldsymbol \sigma^{(l)} + \boldsymbol \phi^{(l)}$ stream via addition before applying the SSM (see Fig.~\ref{fig:COBALTPractical}).

\section{Benchmarking}

We consider the one-dimensional long-range transverse-field Ising model (TFIM)~\cite{koffel_entanglement_2012,defenu_Long-range_2023}
\begin{align}
H = -\,\frac{J}{\mathcal N_\alpha(N)} \sum_{i\neq j} \frac{\hat{\sigma}_i^z \hat{\sigma}_j^z}{r_{ij}^{\alpha}} - h \sum_i \hat{\sigma}_i^x.
\end{align}

\noindent where $\alpha>0$ controls the decay of the Ising interaction, $r_{ij}=\min(|i-j|,\,N-|i-j|)$ is the periodic distance, and $\mathcal N_\alpha(N)$ is a Kac factor that keeps the energy density well-defined across $\alpha$ \cite{defenu_Long-range_2023}. We use $\mathcal N_\alpha(N)=2\sum_{r=1}^{\lfloor N/2\rfloor} r^{-\alpha}$. 
For $J<0$ the low-field phase is ferromagnetic; for $J>0$ it is antiferromagnetic.
In the limit $\alpha\to\infty$ the model reduces to the exactly solvable nearest-neighbor TFIM, diagonalized via the Jordan--Wigner transformation \cite{jordan_uber_1928,Lieb_Two_1961,Pfeuty_One-dimensional_1970}.
At $\alpha=0$ it becomes the fully connected Lipkin--Meshkov--Glick limit \cite{Lipkin_Validity_1965}.

The long-range TFIM spans a diverse set of ground-state behaviors with clean, controllable tuning via $\alpha$:
(i) for $\alpha>3$ short-range Ising criticality with universal exponents;
(ii) for $5/3 < \alpha < 3$ a genuinely long-range regime with continuously varying, $\alpha$-dependent exponents;
and (iii) for $\alpha < 5/3$ a LR mean-field regime, including dangerous irrelevant variables in finite-size scaling \cite{koziol_quantum-critical_2021}.
Away from criticality, long-range couplings generate algebraic (“hybrid exponential–power-law”) correlation tails even in gapped phases, offering a stringent test for variational ansätze to capture both short- and long-distance structure within one architecture ~\cite{roca-jerat_transformer_2024, vodola_long-range_2016}.

\begin{figure}
    \centering
    \includegraphics[width=\linewidth]{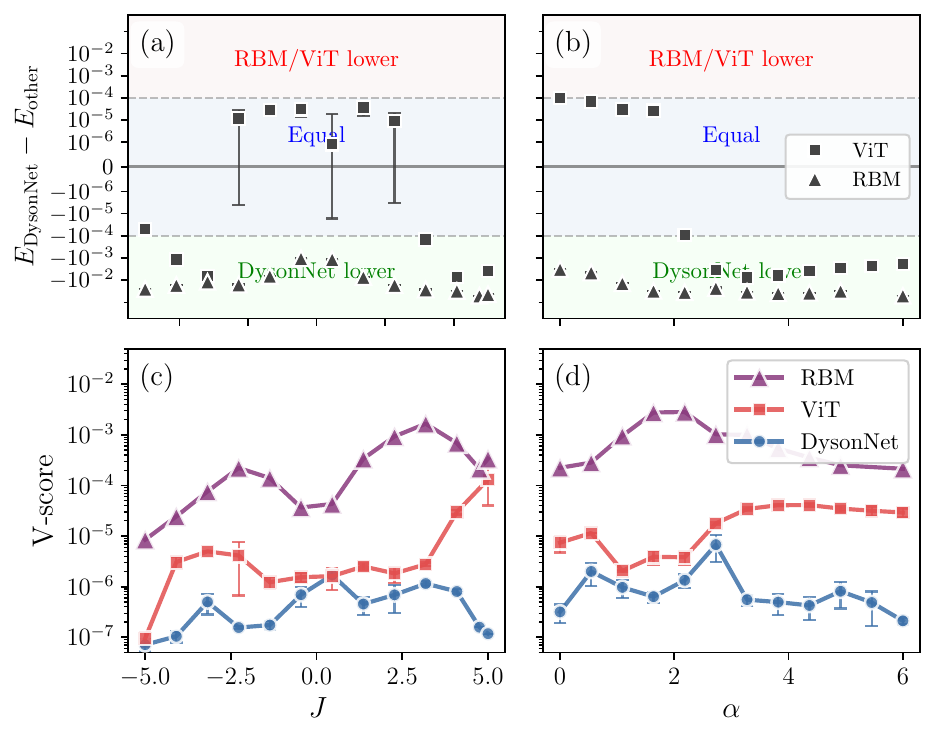}
    \caption{\textbf{\cobalt\  outperforms ViT in ordered phases and on V-score.}
(a–d) Performance of \cobalt\  versus RBM and ViT neural-quantum states (NQS) on the long-range TFIM. 
(a,b) Energy difference \(E_{\mathrm{\cobalt}}-E_{\mathrm{ViT}}\) (squares) and \(E_{\mathrm{\cobalt}}-E_{\mathrm{RBM}}\) (triangles); (c,d) V-score. 
Left: sweep in coupling \(J\) at fixed \(\alpha=4\); right: sweep in \(\alpha\) at fixed \(J=4.75\).
\cobalt\  attains substantially lower energies in the ordered (FM/AFM) regimes by \(10^{-2}\!-\!10^{-3}\), while in the paramagnet ViT is slightly lower (\(<10^{-4}\), effectively equal). \cobalt\  consistently attains lower energies than RBM by at least two orders of magnitude. 
\cobalt\  yields markedly smaller V-scores, especially in the short-range AFM regime (by up to two orders of magnitude vs ViT and four orders of magnitude vs RBM).
Points are means over three training runs with identical sampler settings; error bars denote one s.d.; each model is trained for 400 iterations. System size $N=150$.
}
    \label{fig:ViTvCOBALT}
\end{figure}

We train all models by minimizing the ground-state energy with Variational Monte-Carlo using stochastic reconfiguration (natural-gradient descent). At each iteration we draw Metropolis samples of spin configurations \(\boldsymbol \sigma\) (ABACUS provides \(O(1)\) per-flip amplitudes), estimate the local energy \(E_{\mathrm{loc}}(\boldsymbol \sigma) := (H\Psi)(\boldsymbol \sigma)/\Psi(\boldsymbol \sigma)\) and compute log-derivatives. Parameters are updated as \(\boldsymbol{\theta}\leftarrow\boldsymbol{\theta}+\eta\,\Delta\boldsymbol{\theta}\). Unless noted otherwise, sampler settings, iteration budgets, and solver tolerances are matched across models; detailed training settings are summarized in Appendix Sec.~\ref{app:training_details}.

\begin{figure}
    \centering
    \includegraphics[width=\linewidth]{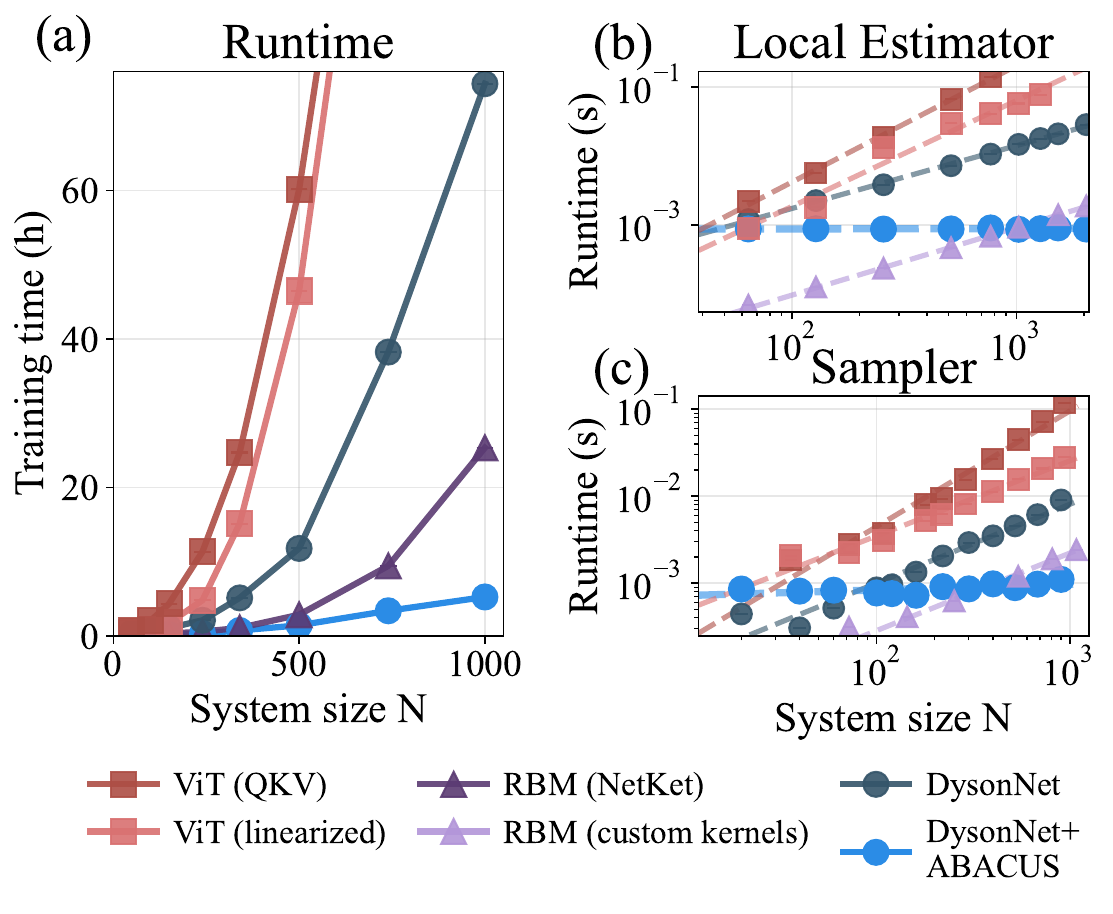}
\caption{\textbf{ABACUS enables constant-time local updates and improves asymptotic scaling.}
(a) Total training time for 400 iterations with Monte-Carlo sweep length \(N/50\), using for each method the hyperparameters that give best overall performance (same choices as in Fig.~\ref{fig:ViTvCOBALT}); at \(N=500\), ViT requires \(\sim60\) h while \cobalt+ABACUS finishes in \(\sim2.5\) h. 
(b) Local-estimator runtime per connected matrix element with \emph{matched} hyperparameters across models (Table~\ref{tab:runtime_matched_hparams}); \cobalt+ABACUS is approximately flat in \(N\) and at \(N=1000\) is \(230\times\) faster than ViT and \(16\times\) faster than \cobalt. 
(c) Single-spin-flip proposal time under the same matched setting; only ABACUS remains approximately constant in \(N\), giving a \(99\times\) speedup vs ViT and \(8.3\times\) vs \cobalt\  at \(N=1000\). RBM baselines have smaller constants than ViT at moderate \(N\) but scale worse than ABACUS-accelerated updates. 
All results use \(J=-5\), \(\alpha=6\), 1024 samples, TF32 precision, and an NVIDIA Tesla P100.}

    \label{fig:kernel_runtime}
\end{figure}

We benchmark \cobalt\ versus ViT and RBM neural quantum states on the long–range TFIM with system size \(N=150\). ViT represents a highly expressive but expensive Ansatz while RBMs are computationally cheap but less expressive.   Training is matched: identical sampler settings, the same system parameters, and an equal budget of 400 iterations. For ViT we adopt the architecture and hyperparameters reported to yield the strongest prior results \cite{roca-jerat_transformer_2024}; for \cobalt\  we optimize architecture parameters independently while keeping the training protocol identical. We perform two sweeps, \(J\) at fixed \(\alpha=4\) and \(\alpha\) at fixed \(J=4.75\). Each configuration is trained three times with independent initializations and sampler seeds; we report the mean and one standard deviation across three training runs. As a secondary accuracy metric we use the V–score. The V-score, defined as $N\mathrm{Var}(E)/E^2$, is a system-size-independent measure of how close the learned state is to an eigenstate. Additional architecture and sampler settings are listed in Appendix Sec.~\ref{app:training_details}.

\cobalt\  achieves substantially lower ground–state energies in the ordered FM and AFM regimes, by \(10^{-2}\) to \(10^{-3}\) against the ViT. In the paramagnet, ViT is  slightly lower, by less than \(10^{-4}\), a scale comparable to typical small–system deviations from the exact energy, so the two are effectively equal. \cobalt\  consistently outperforms the RBM by at least two orders of magnitude.  \cobalt\  attains consistently smaller V–scores across all parameters, with the largest gains in the short–range AFM regime where the scores improve by up to two orders of magnitude and four orders of magnitude versus the RBM.

We compare wall-clock costs for RBM, ViT, linearized ViT~\cite{roca-jerat_transformer_2024}, \cobalt, and \cobalt+ABACUS on the TFIM using the protocol of Fig.~\ref{fig:ViTvCOBALT}. Because these architectures exhibit different accuracy–compute tradeoffs, panel (a) reports end-to-end training time using the best-performing hyperparameters for each method, i.e., the practical cost of obtaining the best accuracy. Panels (b) and (c) isolate the runtime of the local-estimator and sampler-proposal kernels; matched hyperparameters for this comparison are provided in Appendix Table~\ref{tab:runtime_matched_hparams}. For fairness, all compared methods are run in TF32 precision and without custom CUDA kernels for ABACUS. Regarding the RBM, Panel (a) displays the NetKet implementation; however, NetKet prioritizes flexibility and does not use specialized local updates for RBMs. Therefore, to test \cobalt+ABACUS against the theoretical limit of RBM performance, we compare the sampler and estimator against custom RBM kernels (violet) that leverage local updates.


As shown in panel (a) of Fig.~\ref{fig:kernel_runtime}, \cobalt+ABACUS is already strongly advantageous at moderate system sizes: at $N=500$ spins it trains in 1.5 h, compared to 11.8 h for the directly comparable \cobalt\  baseline (same model class without ABACUS) and $46.5$ or $60.1$ h for the ViT variants. It is also about $2\times$ faster than the RBM baseline 2.9 h at $N=500$ . We note that our RBM uses NetKet’s general purpose implementation and is not optimized for the best possible constant factors. At $N=1000$ the gap to the NetKet RBM becomes substantially larger due to different asymptotic scaling, with \cobalt+ABACUS finishing in 5.29 h versus 25.4 h for the RBM (about a $5\times$ difference), while \cobalt\  without ABACUS grows to 74.3 h. Importantly, the RBM is not competitive in accuracy. In our setting—its V-score is worse by $\sim 4$ orders of magnitude—so it is both markedly less accurate and, at larger $N$, slower in wall-clock time.

Panels (b,c) isolate the cost of the two kernels that dominate VMC wall time: evaluating connected matrix elements (local estimator) and proposing single-spin updates. In both cases, \cobalt+ABACUS is the only method with $O(1)$ cost approximately independent of system size $N$. By contrast, \cobalt\  and RBM updates scale roughly as $O(N)$, while ViT-style attention scales as $O(N^{2})$ over the range shown. At $N=1000$, ABACUS yields a $16\times$ (estimator) and $8\times$ (sampler) speedup over the most direct baseline, \cobalt\  without ABACUS; relative to ViT, the gains are $230\times$ and $\sim 100\times$, respectively. For small systems the RBM kernels are faster because their constants are excellent (one matvec plus a nonlinearity per update), but the linear scaling is eventually overtaken: for sampling, \cobalt+ABACUS becomes faster for $N \gtrsim 340$, and for the estimator it becomes more efficient for $N \gtrsim 1000$. Since the RBM is also $\sim 4$ orders of magnitude worse in V-score in our setting, \cobalt+ABACUS can outperform this fast-but-inaccurate baseline in both speed and accuracy at larger $N$. Absolute speedup factors can shift with hyperparameters and implementations, but the scaling advantage is robust and translates into substantial gains in realistic regimes.

We investigate whether the \cobalt\  Ansatz accurately captures ground state correlations in large systems where comparison to ED or ViT is impossible. Using ABACUS, we trained \cobalt\  for $\alpha = 1.5, 2.5,$ and $6$ in the ferromagnetic phase and extracted critical exponents from finite-size-scaling collapses with PyFSSA~\cite{sorge_pyfssa_2015}. Table \ref{tab:crit-exp-compare} lists the extracted critical exponents, and Appendix Figure \ref{fig:collapseTripple} displays the data collapses. Uncertainties represent the maximum values derived from bootstrap and jackknife methods. We achieve system sizes of $N=340$ ($\alpha =6$), $N = 500$ ($\alpha=2.5$), and $N=1000$ ($\alpha =1.5$), surpassing the limit of $N=150$ accessible to ViT \cite{roca-jerat_transformer_2024}. Full training protocol details are provided in Appendix Sec.~\ref{app:training_details}.

Table~\ref{tab:crit-exp-compare} compares the extracted critical points and exponents with exact Stochastic Series Expansion (SSE) benchmarks and with prior Vision Transformer (ViT) results. See also Fig. \ref{fig:collapseTripple} in the appendix for the collapse plots. \cobalt\ reproduces the SSE values across all phases. The only appreciable discrepancy is the critical coupling \(J_c\) at \(\alpha = 1.5\), where our estimate agrees with ViT rather than SSE. This difference likely arises from either generic biases of Neural Quantum States or a minor mismatch in Hamiltonian normalization. In any case, it is not a concern: \(J_c\) is non-universal, whereas the universal critical exponents remain correct.

We obtained these results without tuning the nonlinearity width (\(W=5\)), network depth (\(L=2\)), or proposal spacing (\(s=10\)). Holding these parameters fixed indicates that they do not need to scale with system size, even at criticality. As a result, ABACUS supports constant-time local updates for both the sampler and the local estimator. Most importantly, \cobalt\  improves the ViT estimate of the correlation-length exponent \(\nu\) at \(\alpha = 1.5\), yielding \(\nu \approx 2.04\) instead of \(\nu \approx 1.6\). This gain in precision comes from our ability to study systems up to \(N=1000\), well beyond the \(N \approx 200\) sizes accessible in earlier work~\cite{roca-jerat_transformer_2024}.

To demonstrate that DysonNet+ABACUS is not tailored to the long-range TFIM, we benchmark the same architecture and local-update machinery in Appendix~\ref{app:j1j2} on the frustrated antiferromagnetic $J_1$--$J_2$ Heisenberg chain, a standard test for variational wave-function ans\"atze. This model spans a critical Luttinger-liquid regime and a gapped, spontaneously dimerized phase seperated by a Berezinskii--Kosterlitz--Thouless transition, continuously connecting to the exactly dimerized Majumdar--Ghosh point~\cite{okamoto_fluid-dimer_1992, eggert_numerical_1996, white_dimerization_1996}. In the ED-accessible setting $N=20$, DysonNet maintains high accuracy across the full sweep achieving an energy error of $10^{-4}$ and a V-score of $10^{-3}$.

\begin{table}
\centering
\caption{Critical points and exponents for the long-range TFIM.}
\label{tab:crit-exp-compare}
\setlength{\tabcolsep}{7pt}
\begin{tabular}{c l c c c}
\toprule
\multicolumn{5}{c}{\textbf{FM}}\\
\midrule
$\alpha$ & Method & $J_c$ & $\nu$ & $\beta$ \\
\midrule
1.5 & ViT (Ref.~\cite{roca-jerat_transformer_2024})  & $-1.27(3)$ & $1.6(5)$ & $0.46(7)$ \\
1.5 & SSE (Ref.~\cite{koziol_quantum-critical_2021})                         & $-1.2107$ & $1.9911$ & $0.4968$ \\
1.5 & Theory                                       & ---        & $2.0$    & $0.5$    \\
\addlinespace
1.5 & \textbf{\cobalt}                     &            -1.302(7) &    2.04(6)     &     0.51(1)             \\
\midrule
2.5 & ViT (Ref.~\cite{roca-jerat_transformer_2024}) & $-2.09(1)$ & $1.08(9)$ & $0.19(3)$ \\
2.5 & SSE (Ref.~\cite{koziol_quantum-critical_2021})                         & $-2.0878$ & $1.1084$ & $0.1880$ \\
\addlinespace
2.5 & \textbf{\cobalt}                     &          -2.104(1) &    1.10(4)     &     0.202(2)      \\
\midrule
6.0 & ViT (Ref.~\cite{roca-jerat_transformer_2024}) & $-2.963(1)$ & $1.00(3)$ & $0.12(3)$ \\
6.0 & SSE (Ref.~\cite{koziol_quantum-critical_2021})                         & $-2.9444$  & $0.9849$  & $0.1238$ \\
6.0 & Theory                                       & ---        & $1.0$     & $0.125$  \\
\addlinespace
6.0 & \textbf{\cobalt}                     &          -2.94(1)  &     1.01(7)     &  0.13(4)         \\
\bottomrule
\end{tabular}
\end{table}

\section{Conclusion}

We introduce ABACUS, an efficient local update algorithm, and \cobalt, a physically interpretable architecture. ABACUS updates the wavefunction amplitude for any NQS with global linear layers and strictly local non-linearities in $\mathcal O(1)$ time given cached link tensors. Designed for ABACUS, the architecture \cobalt\ constructs wavefunction amplitudes from local non-linearities connected by global convolutions. It is interpretable as a truncated Dyson series and allows efficient link tensor construction. To amortize caching costs, we employ a screened typewriter sampler that preserves tensors over multiple flips. We derive worst-case runtime bounds and demonstrate a linear speedup. Benchmarks against the long-range transverse-field Ising and frustrated $J_1-J_2$ models show that \cobalt\ matches Vision Transformer accuracy while reducing runtime by two orders of magnitude—an $\mathcal O(N^2)$ asymptotic improvement. Finally, critical exponent calculations confirm that \cobalt\  captures correlations at criticality and scales to large systems.

The ABACUS update is straightforwardly compatible with advanced Monte-Carlo strategies, including cluster moves and parallel tempering, opening the door to further runtime improvements \cite{Swendsen_Nonuniversal_1987, Earl_Parallel_2005, Smith_Optimizing_2025}.
Because the link-tensor construction is agnostic to lattice geometry, \cobalt\ and ABACUS extend naturally to 2D, inviting applications to the 2D systems ~\cite{carleo_solving_2017,lange_architectures_2024,gu_solving_2025,qian_describing_2025}. The same structure also accommodates dynamics: \cobalt\  is compatible with different approaches to time evolution such as the time dependent variational principle~\cite{carleo_solving_2017,schmitt_quantum_2020, schmitt_simulating_2025, reh_time-dependent_2021,nys_ab-initio_2024} or explicitly time dependent NQS~\cite{van_de_walle_many-body_2025,vovk_neural_2025,Romero-Ros_Quantum_2025}.
 Developing efficient algorithms for link tensor construction beyond the translationally invariant \cobalt\  case could yield significant efficiency gains for linearized attention in frustrated or disordered phases \cite{halko_finding_2011, batselier_computing_2018, rende_foundation_2025}.

Beyond \cobalt, this work points to a general NQS design principle: build long-range correlations by wiring strictly local nonlinearities through global linear layers. This might preserve the expressivity of deep networks while adding algebraic structure that ABACUS can exploit for efficient local updates, pushing variational Monte Carlo to larger system sizes. At the same time, the resulting ansätze admit a diagrammatic interpretation as a scattering series, offering a route to theoretical understanding and characterization via field theory and the renormalization group. We believe this offers a path towards architectures where computational efficiency is enabled by physical interpretability.

\begin{acknowledgments}
We thank A.~Bohrdt, N.~Defenu, M.~Heyl, J.~Knolle, D.~Malz, Y.~Minoguchi, and C.~Wanjura for fruitful discussions.
This research was funded in whole or in part by the Austrian Science Fund (FWF) [10.55776/COE1].
The computational results presented have been achieved using the Austrian Scientific Computing (ASC) infrastructure (formerly Vienna Scientific Cluster, VSC). Our NQS simulations were performed using NetKet~\cite{vicentini_netket_2022,carleo_netket_2019}, which is built on top of JAX~\cite{bradbury_jax_2018}, Flax~\cite{heek_flax_2024}, and Optax~\cite{deepmind_optax_2020}.
For Open Access purposes, the authors have applied a CC-BY public copyright license to any author accepted manuscript version arising from this submission.
\end{acknowledgments}

\section*{CODE AVAILABILITY}
We provide a reference implementation of \emph{DysonNet} and ABACUS including the core model components and update algorithm used in this work at \url{https://github.com/lucas-winter/DysonNet.git}. 

\bibliography{citations}

\newpage
\appendix

\section{\cobalt\  can exactly represent power law decay}
\label{app:power-law}

For interpreting the \cobalt\  Ansatz and showing it can express power law decay we will consider the simplifying case of one layer, activation function $f(x) = e^x$ and embedding dimension $d=1$. The log amplitude of the network is 
\begin{align}
\log\Psi(\boldsymbol \sigma)
=\frac{A}{N}\sum_{i,j} D(\sigma_{i-w:i+w})\,G_{i-j}\sigma_j ,
\label{eq:logpsi_short}
\end{align}

\noindent To interpret the local vertex we will apply a Fourier-Walsh expansion. Let the window be $\mathcal W=\{-w,\dots,w\}$ with $W=2w+1$.
Any function $D:\{\pm1\}^{W}\to\mathbb R$ has the exact Walsh expansion

\begin{align}
D(\sigma_{i-w:i+w})
&=\sum_{S\subseteq\mathcal W} D_{S}\,
\prod_{u\in S}\sigma_{i+u}
\nonumber\\[-2pt]
&=D_{\varnothing}
+\sum_{u\in\mathcal W} D_{\{u\}}\sigma_{i+u}
\label{eq:walsh_full_short}\\&+\sum_{u<v} D_{\{u,v\}}\sigma_{i+u}\sigma_{i+v}
+\cdots\  .\nonumber
\end{align}

 \noindent Inserting \eqref{eq:walsh_full_short} into \eqref{eq:logpsi_short} gives
\begin{align}
\log\Psi
=\frac{A}{N}\sum_{S\subseteq\mathcal W}\sum_{I,j}
D_S\,G_{I-j}\,
\bigg(\prod_{u\in S}\sigma_{i+u}\bigg)\,\sigma_j,
\label{eq:coupling_order_short}
\end{align}
\noindent so a Walsh monomial of degree $|S|$ contributes an interaction of degree
$|S|+1$ in the exponent (coupling order).

\noindent We first consider the Jastrow limit for $W = 1$. 
For $W=1$ ($w=0$) only $S=\varnothing$ and $S=\{0\}$ appear, hence
$D(\sigma_i)=D_0+D_1\sigma_i$. Then \eqref{eq:logpsi_short} reduces to
\begin{align}
\log\Psi
&=\frac{A}{N}\sum_I(D_0+D_1\sigma_i)\sum_j G_{i-j}\sigma_j \nonumber\\[-2pt]
&=h\sum_i\sigma_i+\sum_{i<j}J_{ij}\sigma_i\sigma_j+\text{const},
\label{eq:jastrow_short}
\end{align}
where (for even $G_r$ and absorbing $i=j$ into the constant)
\begin{align}
h=\frac{A D_0}{N}\sum_r G_r,
\qquad
J_{ij}=\frac{2A D_1}{N}\,G_{|i-j|}\ . 
\label{eq:Jij_short}
\end{align}

\noindent\paragraph*{Configuration-sum correlator and small-$\lambda$ expansion.}
For the $W=1$ Jastrow limit, diagonal correlators are computed as explicit
configuration sums with respect to the induced classical measure
$p(\boldsymbol \sigma)\propto|\Psi(\boldsymbol \sigma)|^2$,
\begin{align}
\langle\hat{\sigma}_i^z\hat{\sigma}_{i+r}^z\rangle
&=\frac{1}{Z}\sum_{\boldsymbol \sigma}\sigma_i\sigma_{i+r}\,
\exp\!\Big[\,2h\sum_\ell\sigma_\ell
+2\sum_{\ell<m}J_{\ell m}\sigma_\ell\sigma_m\Big],
\label{eq:corr_sum_fleshed}\\[-2pt]
Z
&=\sum_{\boldsymbol \sigma}
\exp\!\Big[\,2h\sum_\ell\sigma_\ell
+2\sum_{\ell<m}J_{\ell m}\sigma_\ell\sigma_m\Big].
\nonumber
\end{align}
To expose the long-distance scaling we set $h=0$ (or, equivalently, consider the
connected correlator), write $J_{\ell m}=\lambda\,\bar J_{\ell m}$, and expand
around $\lambda=0$, where the weight becomes uniform.
Denote averages with respect to the uniform measure by $\langle\cdot\rangle_0$,
i.e. $\langle F\rangle_0:=2^{-N}\sum_{\boldsymbol \sigma}F(\boldsymbol \sigma)$.
Then
\begin{align}
\langle\sigma_i\sigma_{i+r}\rangle
&=\frac{\Big\langle \sigma_i\sigma_{i+r}\,
\exp\!\big[2\lambda\sum_{\ell<m}\bar J_{\ell m}\sigma_\ell\sigma_m\big]
\Big\rangle_0}
{\Big\langle
\exp\!\big[2\lambda\sum_{\ell<m}\bar J_{\ell m}\sigma_\ell\sigma_m\big]
\Big\rangle_0}
\label{eq:numden_expand}\\
&=\frac{
\Big\langle \sigma_i\sigma_{i+r}\Big\rangle_0
+2\lambda\sum_{\ell<m}\bar J_{\ell m}
\Big\langle \sigma_i\sigma_{i+r}\sigma_\ell\sigma_m\Big\rangle_0
+O(\lambda^2)}
{
1+2\lambda\sum_{\ell<m}\bar J_{\ell m}
\Big\langle \sigma_\ell\sigma_m\Big\rangle_0
+O(\lambda^2)
}.
\nonumber
\end{align}

where we expanded the exponential up to second order in the second line. Since under $\langle\cdot\rangle_0$ spins are independent with
$\langle\sigma_a\rangle_0=0$ and $\langle\sigma_a\sigma_b\rangle_0=\delta_{ab}$.
For $r\neq0$ one also has $\langle\sigma_i\sigma_{i+r}\rangle_0=0$.
Moreover, the only non-vanishing four-spin averages are those where the indices
pair up; in particular,
$\langle \sigma_i\sigma_{i+r}\sigma_\ell\sigma_m\rangle_0=1$
iff $\{\ell,m\}=\{i,i+r\}$ (and vanishes otherwise). Hence the sum in
\eqref{eq:numden_expand} collapses to the single edge $(\ell,m)=(i,i+r)$.
\begin{align}
\langle\sigma_i\sigma_{i+r}\rangle
=2\lambda\,\bar J_{i,i+r}+O(\lambda^2)
=2J_{i,i+r}+O(J^2).
\label{eq:corr_first_order}
\end{align}
Using the $W=1$ identification $J_{ij}\propto G_{|i-j|}$ and a power-law kernel
$G_r=r^{-\alpha}$ therefore yields the leading asymptotic behavior
\begin{align}
\langle {\sigma}_i {\sigma}_{i+r}\rangle
=2J_r+O(J^2)
\sim r^{-\alpha}+O(J^2)
.
\label{eq:corr_powerlaw_fleshed}
\end{align}
Higher orders in $\lambda$ correspond to multi-edge contributions (paths and
graphs) and dress the prefactor, while the leading long-distance tail is set by
the direct long-range coupling inherited from $G_r$.

For $W>1$, higher Walsh sectors $D_S$ in \eqref{eq:walsh_full_short}
generate controlled multi-spin vertices of order $|S|{+}1$
via \eqref{eq:coupling_order_short}, adding non-Gaussian short-range dressing
beyond a pure Jastrow factor. Stacking $L$ \cobalt\  blocks increases the
effective coupling order through the Dyson-series structure of repeated
$D(\boldsymbol \sigma)G$ insertions, while preserving the separation of scales:
$G$ fixes the long-distance response (e.g.\ algebraic kernels suitable for
criticality) and $D$ learns non-universal local structure. Allowing complex
parameters further enables nontrivial sign/phase structure beyond strictly
positive Jastrow states.

\section{Induction Proof of ABACUS}
\label{sec:induction_proof} 
In this section we prove Theorem \ref{thm:abacus} and the recursion relations
\eqref{eq:recur_step} and \eqref{eq:recur_out} via induction.
For a local update \(\boldsymbol \sigma=\boldsymbol \sigma_0+\Delta\boldsymbol \sigma\), define
\begin{align}
\mathcal M(\boldsymbol \sigma)
&:=\prod_{l=1}^{L}\!\bigl(D^{(l)}_0+\Delta D^{(l)}\bigr)M^{(l)},
\\
\boldsymbol \Omega(\boldsymbol \sigma)
&=A\,P_M\,\mathcal M(\boldsymbol \sigma)\,\boldsymbol \sigma .
\end{align}

A \emph{matrix word} is one summand in the expansion of \(\mathcal M(\boldsymbol \sigma)\).
Its \emph{pivot} is the largest layer index \(l\) for which \(\Delta D^{(l)}\) appears.
Let \(\mathcal W_l\) be the set of words with pivot \(l\), and define
\begin{equation}
\tilde{\boldsymbol h}^{(l)}
:=P_{S,j}\sum_{W\in\mathcal W_l}W\,\boldsymbol \sigma
\quad\in\mathbb R^{W\times d}.
\end{equation}

\begin{invariant}[after line 3 of Algorithm~\ref{alg:local-update-short}]
For every \(l\le L\), the recursion \eqref{eq:recur_step} equals the sum over all words with pivot \(l\).
\begin{equation}
\boldsymbol{\tilde h}^{(l)}
=P_{S,j}\sum_{W\in\mathcal W_l}W\,\boldsymbol \sigma\;.
\label{eq:inv}
\end{equation}
\end{invariant}

\noindent \emph{Proof.} We use strong induction over \(l\).

\smallskip
\emph{Base case \(l=1\).}

For pivot \(1\), the only scattering layer is the first one.
\begin{align}
\tilde{\boldsymbol h}^{(1)}
&=P_{S,j}\,\Delta D^{(1)}M^{(1)}(\boldsymbol \sigma_0+\Delta\boldsymbol \sigma)
\nonumber\\
&=\Delta D^{(1)}\!\left(M^{(1)}\Delta\boldsymbol \sigma+\tilde{\boldsymbol h}^{(1)}_0\right),
\end{align}
which is exactly the initialization in Algorithm~\ref{alg:local-update-short}.

\smallskip
\emph{Induction step.}
Assume \eqref{eq:inv} holds for all \(k<l\). Any word with pivot \(l\) is obtained by
applying \(\Delta D^{(l)}M^{(l)}\) to either:
(i) the unscattered incident field at layer \(l\), or
(ii) a word with pivot \(m<l\), propagated from \(m\) to \(l\) through background factors
\(D_0^{(n)}M^{(n)}\), \(n=m+1,\dots,l-1\).
Therefore,
\begin{align}
\sum_{W\in\mathcal W_l}W\boldsymbol \sigma
&=
\Delta D^{(l)}M^{(l)}
\bigg[
\boldsymbol h^{(l)}_0
 \\&+\sum_{m=1}^{l-1}
\left(\prod_{n=m+1}^{l-1}D_0^{(n)}M^{(n)}\right)
\sum_{W\in\mathcal W_m}W\boldsymbol \sigma
\bigg]\nonumber.
\end{align}
Projecting to the slice and using the induction hypothesis gives
\begin{align}
\tilde{\boldsymbol h}^{(l)}
=
\Delta D^{(l)}\!\left(
\tilde M^{(l)}\tilde{\boldsymbol h}^{(l-1)}
+\sum_{m=1}^{l-2}L^{(l,m)}\tilde{\boldsymbol h}^{(m)}
+\tilde{\boldsymbol h}_0^{(l)}
\right),
\end{align}
where \(L^{(l,m)}\) is exactly the link tensor from Eq.~\eqref{eq:Lmn}.
\begin{align}
L^{(l,m)}
=
P_{S,j}M^{(l)}
\left(\prod_{n=m+1}^{l-1}D_0^{(n)}M^{(n)}\right)P_{S,j}^{\top}.
\end{align}
This is Eq.~\eqref{eq:recur_step}, so the invariant is preserved.

\emph{Termination:} After the loop (\(l=L\)), \eqref{eq:inv} holds for every pivot class.
Grouping the full expansion by pivot and contracting each class with its output link
tensor \(T^{(l)}\) yields
\begin{align}
\boldsymbol \Omega(\boldsymbol \sigma)
=
A\!\left(P_M\,\tilde{\boldsymbol h}^{(L)}
+ \sum_{l=1}^{L}T^{(l)}\tilde{\boldsymbol h}^{(l)}\right),
\end{align}
which is Eq.~\eqref{eq:recur_out}. Hence ABACUS produces the desired matrix-vector
product exactly.

\emph{Complexity.} Using only the principal blocks of the Toeplitz matrix, contracting a vector $\boldsymbol{\tilde h}^{(l)}$ with $G^{(l)}$ and then applying the slice projection $P_{S,j}$ requires $\mathcal O\!\left(d\,W^{2}\right)$ operations, where $W = 2w+1 \ll N$ denotes the slice width and $d$ is the hidden‐state dimension.  If $W$ becomes large, the same contraction can instead be performed with a state-space model (SSM) in $\mathcal O(W\,d\,s^{2})$ time, which scales strictly linearly in $W$.

Multiplying $\boldsymbol{\tilde h}^{(l)}$ by the sparse diagonal update $\Delta D^{(l)}$ costs $\mathcal O\!\left(W\,d^{2}\right)$, again linear in the slice width.  The link tensors at the final layer, $T^{(l)}$, incur the same $\mathcal O\!\left(W\,d^{2}\right)$ cost, while the most expensive step is contracting with the multi-layer link tensors, which scales as $\mathcal O\!\left(W^{2}\,d^{2}\right)$ yet remains independent of the full system size $N$.

Since every individual contraction is free of any $N$-dependence, the entire update can be carried out in constant time with respect to $N$.  For a network of $L$ layers, roughly $L(L-1)/2 \sim \mathcal O(L^{2})$ such contractions are required, giving the overall worst-case complexity $\mathcal O\!\bigl(d^{2}W^{2}L^{2}\bigr),$
which is constant in $N$ and quadratic in both the number of layers and the slice width.

By a strong-induction loop-invariant argument we have shown that
Algorithm~\ref{alg:local-update-short} constructs every matrix word
exactly once, associates it with the correct pre-contracted linear
environment, and thus evaluates the network output \(\boldsymbol \Omega(\boldsymbol \sigma)\)
at \(\mathcal O(d^2W^2 L^2)\) constant in $N$ cost for a local spin flip.
\hfill\(\blacksquare\)

\section{Screened Metropolis sampler is exact} \label{app:screenedProof}
Let $\pi(\boldsymbol \sigma)\propto |\Psi(\boldsymbol \sigma)|^2$ be the target distribution and consider a symmetric proposal
$q(\boldsymbol \sigma\to\boldsymbol \sigma')=q(\boldsymbol \sigma'\to\boldsymbol \sigma)$.
Define the exact Metropolis--Hastings ratio
\begin{equation}
R(\boldsymbol \sigma,\boldsymbol \sigma') := \frac{|\Psi(\boldsymbol \sigma')|^2}{|\Psi(\boldsymbol \sigma)|^2},\qquad 
\alpha(\boldsymbol \sigma\to\boldsymbol \sigma') := \min\{1,R(\boldsymbol \sigma,\boldsymbol \sigma')\}.
\end{equation}
Assume we can compute deterministic bounds $L(\boldsymbol \sigma,\boldsymbol \sigma')$ and $U(\boldsymbol \sigma,\boldsymbol \sigma')$ such that
\begin{equation}
\label{eq:LUbound}
L(\boldsymbol \sigma,\boldsymbol \sigma') \le R(\boldsymbol \sigma,\boldsymbol \sigma') \le U(\boldsymbol \sigma,\boldsymbol \sigma') \qquad \text{for all } (\boldsymbol \sigma,\boldsymbol \sigma').
\end{equation}
Consider the screened acceptance rule: draw $u\sim\mathcal U[0,1]$ and
(i) accept if $u<\min\{1,L\}$; 
(ii) reject if $u>\min\{1,U\}$; 
(iii) otherwise (the ambiguous window) compute $R$ exactly and accept iff $u<\min\{1,R\}$.
Then the resulting accept/reject decision coincides with the standard MH decision for every triple
$(\boldsymbol \sigma,\boldsymbol \sigma',u)$, hence the induced Markov chain satisfies detailed balance with respect to $\pi$ and is exact.

\emph{Proof.}
By \eqref{eq:LUbound} we have $\min\{1,L\}\le \min\{1,R\}\le \min\{1,U\}$.
If $u<\min\{1,L\}$ then $u<\min\{1,R\}$ and both screened MH and MH accept.
If $u>\min\{1,U\}$ then $u>\min\{1,R\}$ and both reject.
Otherwise $u\in[\min\{1,L\},\min\{1,U\}]$ and the screened rule evaluates $R$ and applies the MH test
$u<\min\{1,R\}$, which is exactly the MH decision.
Therefore the transition kernel is identical to MH, and detailed balance
$\pi(\boldsymbol \sigma)q(\boldsymbol \sigma\to\boldsymbol \sigma')\alpha(\boldsymbol \sigma\to\boldsymbol \sigma')=\pi(\boldsymbol \sigma')q(\boldsymbol \sigma'\to\boldsymbol \sigma)\alpha(\boldsymbol \sigma'\to\boldsymbol \sigma)$
holds. \qed

\section{Upper bound error scaling for independent scattering approximation} \label{app:metropolis-error}

In this section we will explicitly derive an upper bound for the error in the Metropolis acceptance ratio under the \emph{independent scattering approximation} introduced in section \ref{sec:sampler}.
For a configuration $\boldsymbol \sigma$ and a fixed reference configuration $\boldsymbol \sigma_0$, define
\begin{equation}
  \epsilon(\boldsymbol \sigma)
  := \left|\frac{\Psi_{\mathrm{ISA}}(\boldsymbol \sigma)-\Psi(\boldsymbol \sigma)}{\Psi(\boldsymbol \sigma_0)}\right|.
  \label{eq:acceptance-error}
\end{equation}

Let the wavefunction be of the form
\begin{align}
\begin{split}
  \Psi(\boldsymbol \sigma) &= f\!\big(\vect z(\boldsymbol \sigma)\big),\\
  \vect z(\boldsymbol \sigma) &:= A\,P_M\!\Bigg(G^f\sum_{M=0}^{L}\prod_{l=1}^{M} D^{(l)}(\boldsymbol \sigma)\,G^{(l)}\Bigg)\,E\,\boldsymbol \sigma,
  \label{eq:network}
\end{split}
\end{align}
where $P_M$ is the mean-pooling projection vector, $A\in\mathbb{R}^{d\times d}$, $D^{(l)}\in\mathbb{R}^{d\times d\times N}$ (block diagonal along the input sequence), and $G^{(l)}\in\mathbb{R}^{N\times N\times d}$ (Toeplitz, block diagonal in $d$). 
Fix the reference $\boldsymbol \sigma_0$ and set
\begin{equation}
  \vect z_0 := \vect z(\boldsymbol \sigma_0), 
  \quad
  \Delta \vect z(\boldsymbol \sigma) := \vect z_{\mathrm{ISA}}(\boldsymbol \sigma) - \vect z(\boldsymbol \sigma),
  \quad
  \boldsymbol \sigma = \boldsymbol \sigma_0 + \Delta\boldsymbol \sigma.
  \label{eq:defs-z}
\end{equation}

\noindent with 

\begin{align}
    z_{\text{ISA}}(\boldsymbol \sigma) = z(\boldsymbol \sigma_0) + \sum_k \text{ABACUS}(\Delta \boldsymbol \sigma_k , \boldsymbol \sigma_0, \mathcal T(\boldsymbol \sigma_0)) \ , 
\end{align}

\noindent with the link tensors $\mathcal T( \boldsymbol{\sigma}_0) = \text{LinkTensor}(\boldsymbol \sigma_0)$ and the flips $\Delta \boldsymbol \sigma_k$. Then~\eqref{eq:acceptance-error} can be written as
\begin{equation}
  \epsilon(\boldsymbol \sigma)
  = \frac{\big|\,f(\vect z(\boldsymbol \sigma)+\Delta\vect z(\boldsymbol \sigma)) - f(\vect z(\boldsymbol \sigma))\,\big|}
         {f(\vect z_0)}.
  \label{eq:eps-with-f}
\end{equation}

For the choice
\(
f(\vect z)=\sum_{j=1}^{d}\cosh(z_j)
\)
(with \(f(\vect z_0)>0\)), one can bound the numerator of~\eqref{eq:eps-with-f} without any linearization.
Using the identity
\(
\cosh(u+v)=\cosh(u)\cosh(v)+\sinh(u)\sinh(v)
\)
and the inequalities \(|\sinh(u)|\le \cosh(u)\) and \(\cosh(v)+|\sinh(v)|=e^{|v|}\), we obtain for all \(u,v\in\mathbb{R}\)
\begin{align}
  \big|\cosh(u+v)-\cosh(u)\big|
  &\le \cosh(u)\big(\cosh(v)-1+|\sinh(v)|\big)
  \\ &= \big(e^{|v|}-1\big)\cosh(u).
  \label{eq:cosh-lip-exp}
\end{align}
Applying~\eqref{eq:cosh-lip-exp} componentwise with \(u=z_j(\boldsymbol \sigma)\) and \(v=\Delta z_j(\boldsymbol \sigma)\), and using
\(
e^{|\Delta z_j|}-1 \le e^{\|\Delta\vect z(\boldsymbol \sigma)\|_\infty}-1
\),
yields
\begin{align}
  \epsilon(\boldsymbol \sigma)
  &\le \frac 1 {f(\boldsymbol z_0)} \sum_{j=1}^{d}\big(e^{|\Delta z_j(\boldsymbol \sigma)|}-1\big)\cosh\!\big(z_j(\boldsymbol \sigma)\big)
  \nonumber\\
  &\le \frac 1 {f(\boldsymbol z_0)} \big(e^{\|\Delta\vect z(\boldsymbol \sigma)\|_\infty}-1\big)\sum_{j=1}^{d}\cosh\!\big(z_j(\boldsymbol \sigma)\big)
  \nonumber\\
  &= \big(e^{\|\Delta\vect z(\boldsymbol \sigma)\|_\infty}-1\big)\,\frac {f(\vect z(\boldsymbol \sigma))} {f(\boldsymbol z_0)} .
  \\  &\leq \big(e^{\|\Delta\vect z(\boldsymbol \sigma)\|_2}-1\big)\,\frac {f(\vect z(\boldsymbol \sigma))} {f(\boldsymbol z_0)} .
  \label{eq:eps-exact-bound}
\end{align}

We note this exact result can be for small $||\Delta \boldsymbol z||$ approximated as 
\begin{equation}
\;
  \epsilon(\boldsymbol \sigma) \;\lesssim\; \norm{\Delta\vect z(\boldsymbol \sigma)}_2\frac {f(\vect z(\boldsymbol \sigma))} {f(\boldsymbol z_0)} .
  \;
  \label{eq:key-result}
\end{equation}
\noindent The acceptance error~\eqref{eq:acceptance-error} for the mean-field multi-flip approximation is controlled by the feature-space deviation $\Delta\vect z$:
the linearized estimate~\eqref{eq:key-result} shows that bounding $\norm{\Delta\vect z}_2$ suffices to bound~$\epsilon$.

\subsection*{Scaling of \texorpdfstring{$\norm{\Delta \vect z}_2$}{||Δz||\_2} under separated flips}

We now derive the scaling of $\norm{\Delta \vect z}_2$. Decompose the update into spatially disjoint contributions
\begin{equation}
  \Delta \boldsymbol \sigma = \sum_{j\in F}\Delta \sigma_j,
  \qquad F=\{\text{sites flipped relative to }\boldsymbol \sigma_0\}, 
  \label{eq:sigma-decomp}
\end{equation}
with the total number of flips $M:=|F|$. Assume flips are separated by at least a window $w$,
\begin{equation}
  |i-j|>w\quad\text{for all }i\neq j\in F,
  \label{eq:sep-assumption}
\end{equation}
so that the layerwise perturbations admit a separable decomposition with disjoint spatial support,
\begin{equation}
  \Delta D^{(\ell)}(\boldsymbol \sigma,\boldsymbol \sigma_0)
  = \sum_{j\in F}\Delta D^{(\ell)}_{j}\!\left(\Delta \sigma_j,\boldsymbol \sigma_0\right).
  \label{eq:DeltaD-decomp}
\end{equation}
Writing
\begin{equation}
  \tilde D^{(\ell)}_{\alpha} :=
  \begin{cases}
    D^{(\ell)}_{0}, & \alpha=0,\\[3pt]
    \Delta D^{(\ell)}_{j}, & \alpha=j\in F,
  \end{cases}
  \label{eq:Dtilde}
\end{equation}
we may expand the exact and approximate features as
\begin{align}
  \Delta \boldsymbol z 
  &= A\,P_M \!\!\sum_{\substack{\boldsymbol{\alpha}\in(\{0\}\cup F)^L\\ |\{\alpha_\ell\in F\}|\ge 2}}
     \Bigg(\prod_{\ell=1}^{L}\tilde D^{(\ell)}_{\alpha_\ell}G^{(\ell)}\Bigg) E\,\boldsymbol \sigma.
  \label{eq:z-diff-words}
\end{align}
In \eqref{eq:z-diff-words} we have dropped all ``multi-flip within one word'' terms with a single flip ($|\{\alpha_\ell\in F\}|=1$), i.e.\ we retain only cross-insertion words involving two or more distinct flips.

By submultiplicativity of the spectral norm,
\begin{equation}
  \norm{AB}_2\le \norm{A}_2\,\norm{B}_2,
  \label{eq:submult}
\end{equation}
and, for mean pooling
\begin{align}
  P_M := \frac{1}{N}\,[\,I_d\;I_d\;\cdots\;I_d\,]\in\mathbb{R}^{d\times dN},
  \\
  \norm{P_M}_2=\frac{1}{\sqrt{N}},\quad \norm{\boldsymbol \sigma}_2=\sqrt{N},
  \label{eq:mean-pool}
\end{align}
we obtain the cancellation $\norm{P_M}_2\norm{\boldsymbol \sigma}_2=1$ and thus
\begin{equation}
  \norm{\Delta\vect z}_2
  \le
  \norm{A}_2\;
  \bigg\|
  \sum_{\substack{\boldsymbol{\alpha}\in(\{0\}\cup F)^L\\ |\{\alpha_\ell\in F\}|\ge 2}}
  \Big(\prod_{\ell=1}^{L}\tilde D^{(\ell)}_{\alpha_\ell}G^{(\ell)}\Big)
  \bigg\|_2\;
  \norm{E}_2.
  \label{eq:deltaz-pre}
\end{equation}
We absorb fixed-depth and fixed-channel constants into the following envelopes. For each layer define
\begin{equation}
  d_\ell:=\big\|D^{(\ell)}_{0}\big\|_{2},\quad
  \delta_\ell:=\sup_{j}\big\|\Delta D^{(\ell)}_{j}\big\|_{2},
  \label{eq:layer-bounds}
\end{equation}
Treating each $G^{(\ell)}$ as a Toeplitz/convolution operator, a Young–Schur envelope yields, for $n<\ell$ and $i\neq j$,
\begin{align}
  \big\|\Delta D_i^{(\ell)}\,G^{(\ell)}\,(\cdots)\,G^{(n+1)}\,\Delta D_j^{(n)}\big\|_{2}
  \;\le\;
  \delta_\ell\,\delta_{n}\;C^{(n\to\ell)}\;\kappa(|i-j|),
  \label{eq:YS-envelope}
\end{align}
where $\kappa$ is an effective two-point decay kernel in the same class as the $g_\ell$ (exponential or power law), and $C^{(n\to\ell)}$ collects intermediate operator norms between the two insertions.

Absorb all layer bookkeeping into the finite constant (for fixed depth $L$)
\begin{equation}
  K_2
  := \norm{A}_2\,\norm{E}_2
     \sum_{1\le n<\ell\le L}
     C^{(n\to\ell)}\,\delta_\ell\,\delta_n
     \prod_{m\notin[n,\ell]} d_m.
  \label{eq:K2-def}
\end{equation}
Then the contribution from \emph{two-insertion} cross terms obeys
\begin{equation}
\;
  \big\|\Delta \vect z\big\|_{2}^{(r=2)}
  \;\le\;
  K_2 \sum_{\substack{i,j\in F\\ i\neq j}}\kappa\!\big(|i-j|\big).
  \;
  \label{eq:pairwise-bound}
\end{equation}
Higher-order insertions ($r\ge 3$) are controlled by higher convolutions of $\kappa$ and extra small factors $\delta_\ell$; for summable kernels (e.g.\ exponential decay or power law with exponent $>1$) or sufficiently large separation $w$ they are subleading compared to \eqref{eq:pairwise-bound}. We therefore focus on \eqref{eq:pairwise-bound} to obtain scaling laws.

\subsubsection*{Counting pairs under the separation assumption}

Let the flips be ordered along the 1D geometry, and assume \eqref{eq:sep-assumption}. A worst-case upper bound is obtained when flips are arranged with minimal allowed spacing $w$, so that
\begin{equation}
  \sum_{\substack{i,j\in F\\ i\neq j}}\kappa(|i-j|)
  \;\le\;
  2\sum_{m=1}^{M-1}(M-m)\,\kappa(m\,w),
  \label{eq:pair-sum-count}
\end{equation}
where $m$ counts the (coarse-grained) separation in units of $w$ and the factor $2$ accounts for ordered pairs $(i,j)$, $i\neq j$.
Combining \eqref{eq:pairwise-bound} and \eqref{eq:pair-sum-count},
\begin{equation}
  \big\|\Delta \vect z\big\|_{2}^{(r=2)}
  \;\le\;
  2K_2 \sum_{m=1}^{M-1}(M-m)\,\kappa(m\,w).
  \label{eq:pair-master}
\end{equation}

\paragraph{Exponential decay.}
If $\kappa(r)\le C_{\exp}\,e^{-r/\xi}$, then
\begin{align}
  \sum_{m=1}^{M-1}(M-m)\,e^{-m w/\xi}
  &= M\!\sum_{m=1}^{M-1}e^{-m w/\xi} - \sum_{m=1}^{M-1}m\,e^{-m w/\xi}
  \nonumber\\
  &\le M\,\frac{e^{-w/\xi}}{1-e^{-w/\xi}},
\end{align}
whence
\begin{equation}
\;
  \big\|\Delta \vect z\big\|_{2}^{(r=2)}
  \;\lesssim\;
  2K_2\,C_{\exp}\,M\,
  \frac{e^{-w/\xi}}{1-e^{-w/\xi}}
  \;=\;\mathcal{O}\!\big(M\,e^{-w/\xi}\big).
  \;
  \label{eq:exp-scaling}
\end{equation}
In particular, for $w\gg \xi$ the pairwise error is exponentially small in $w$ and only linear in the number of flips $M$.

\paragraph{Power-law decay.}
If $\kappa(r)\le C_{\alpha}\,r^{-\alpha}$ for $r\ge w$ with $\alpha>0$, then from \eqref{eq:pair-master}
\begin{equation}
  \big\|\Delta \vect z\big\|_{2}^{(r=2)}
  \;\le\;
  2K_2\,C_{\alpha}\,w^{-\alpha}
  \sum_{m=1}^{M-1}(M-m)\,m^{-\alpha}.
  \label{eq:plaw-master}
\end{equation}
The sum admits the asymptotics (as $M\to\infty$) so that
\begin{equation}
 \;
  \big\|\Delta \vect z\big\|_{2}^{(r=2)}
  \;\lesssim\;
  \begin{cases}
    \displaystyle 2K_2\,C_{\alpha}\,w^{-\alpha}\,M^{\,2-\alpha}, & 0<\alpha<1,\\[10pt]
    \displaystyle 2K_2\,C_{1}\,w^{-1}\,M\,\log M, & \alpha=1,\\[6pt]
    \displaystyle 2K_2\,C_{\alpha}\,\zeta(\alpha)\,w^{-\alpha}\,M, & \alpha>1.
  \end{cases}
  \;
  \label{eq:plaw-scaling}
\end{equation}
Thus, in the non-extensive regime $0<\alpha<1$ the pairwise error grows superlinearly with $M$ as $M^{2-\alpha}$, while for $\alpha\ge 1$ it is at most quasi-linear in $M$ (with a mild $M\log M$ at $\alpha=1$), and suppressed by $w^{\alpha}$.

\paragraph{Remark on higher-order insertions.}
Words with $r\ge 3$ distinct insertions are bounded by constants $K_r$ analogous to \eqref{eq:K2-def} and $r$-fold convolutions of $\kappa$, leading to
\begin{equation}
  \norm{\Delta \vect z}_2
  \;\le\; \sum_{r=2}^{L} K_r\,\mathcal{S}_r(M,w),
  \qquad
  \mathcal{S}_2(M,w)\text{ as in \eqref{eq:pair-master}}.
\end{equation}
For exponentially decaying $\kappa$ the $r$-fold envelopes remain $\mathcal{O}(e^{-w/\xi})$ and are suppressed by extra small factors $\delta_\ell$, rendering $r\ge 3$ subleading. For power laws with $\alpha>1$, $\kappa\in \ell^1(\mathbb{Z})$ and higher convolutions are summable; the pairwise term dominates. In the marginal/non-extensive case $0<\alpha\le 1$, pairwise still controls provided $\max_\ell \delta_\ell$ is small enough or $w$ is scaled with $M$ so that $w^{-\alpha}$ offsets the combinatorics of $r$-tuples.

\section{Sampling Complexity} \label{app:sampling-complexity}

\noindent
We model premature freezing of an approximate Metropolis chain via a step--dependent stop (``hazard") probability (see also Ref. ~\cite{barlow_Statistical_1981, kalbfleisch_Statistical_2002}). The objective is to compute the expected number of successful moves before the chain either freezes or reaches a deterministic cap. At the $k$-th attempted local flip ($k=1,2,\dots$) the chain freezes with probability
\begin{equation}
  p_k \;=\; p_0\,k^\gamma, 
  \qquad \gamma\ge 0,
  \label{eq:hazard-pk}
\end{equation}
otherwise it proceeds to $k{+}1$. Independently of freezing, the chain is capped after
\begin{equation}
  K \;=\; \frac{N}{w} \;=\; \frac{N^{\,1-\beta}}{w_0},
  \qquad w(N)=w_0 N^\beta,\;\; 0\le \beta \le 1,
  \label{eq:cap-K1}
\end{equation}
accepted moves. We denote by $M$ the number of accepted flips before the first of \emph{freeze} or \emph{cap} occurs, and we seek $\mathbb E[M]$.

\subsection*{General hazard results}

\paragraph{Exact expression.}
Survival through step $k$ requires not freezing in any of the steps $1,\dots,k$. Thus
\begin{equation}
  \Pr(M\ge k) \;=\; \prod_{i=1}^{k}\!\big(1-p_0\,i^\gamma\big),
  \qquad k\le K,
  \label{eq:survival}
\end{equation}
and by the tail--sum formula
\begin{equation}
  \;
  \mathbb E[M] \;=\; \sum_{k=1}^{K}\;\prod_{i=1}^{k}\!\big(1-p_0\,i^\gamma\big)\; .
  \label{eq:EM-exact}
\end{equation}
(If $p_0\,i^\gamma\ge 1$ for some $i$, the product---and thus all later terms---vanishes.)

When $p_0 i^\gamma\ll 1$ up to the relevant $i$, we use $\log(1-x)\approx -x$ and
$\sum_{i=1}^{k} i^\gamma\approx k^{\gamma+1}/(\gamma{+}1)$ to obtain
\begin{equation}
  \mathbb E[M] \;\approx\; \sum_{k=1}^{K}\!
  \exp\!\Big(-\frac{p_0}{\gamma+1}k^{\gamma+1}\Big)
  \;\approx\; \int_0^{K} e^{-a t^{\gamma+1}}\,dt,
  \label{eq:EM-small}
\end{equation}

\noindent with $a := p_0/(\gamma+1)$. Define the \emph{hazard length}
\begin{equation}
 \;
  \ell := a^{-1/(\gamma+1)} \;=\;
  \Big(\frac{\gamma+1}{p_0}\Big)^{\!1/(\gamma+1)} \;
  \label{eq:hazard-length}
\end{equation}
at which the cumulative hazard is $O(1)$. Comparing $\ell$ to the cap $K$ in \eqref{eq:cap-K1} yields two regimes:

\emph{Hazard-limited} ($K\gg \ell$): the integral in \eqref{eq:EM-small} saturates,
\begin{equation}
\;
  \mathbb E[M] \;\sim\; C_\gamma\,p_0^{-1/(\gamma+1)},\qquad
  C_\gamma=\frac{\Gamma\!\big(\tfrac{1}{\gamma+1}\big)}{(\gamma+1)^{\gamma/(\gamma+1)}}\;
  \label{eq:EM-hazardGeneral}
\end{equation}

\emph{Cap-limited} ($K\ll \ell$): freezing is negligible before the cap,
\begin{equation}
  \; \mathbb E[M] \;\sim\; K \;=\; \frac{N}{w} \;
  \label{eq:EM-cap1}\end{equation}
The crossover is conveniently expressed as
\begin{equation}
  p_0\,K^{\gamma+1}\;\sim\; 1,
  \label{eq:crossover}
\end{equation}
equivalently $K\sim \ell$ by \eqref{eq:hazard-length}.

\bigskip

\paragraph{Exponentially decaying correlations (area law).}
A standard bound gives an error per window of the form
$\epsilon \lesssim C\,e^{-w/\xi}\,w^{-1}(k/K)$, suggesting a \emph{linear} hazard $\gamma=1$ with
\begin{equation}
  p_k = p_0 k,\qquad p_0 \sim C e^{-w/\xi}.
  \label{eq:area-p0}
\end{equation}

To keep the error in \eqref{eq:area-p0} controlled, we can rescale the minimum spacing logarithmically with $N$ by setting $w(N)=w_0\log N$. Comparing $K$ and $\ell$,
\begin{equation}
\frac{K}{\ell}\;\sim\; \frac{N^{\,1-w_0/(2\xi)}}{\log N}\,,
\end{equation}
so one crosses from hazard- to cap-limited as $w_0$ increases. In particular, for sufficiently large $w_0$ (e.g.\ $w_0>2\xi$ in our conservative criterion), the chain is \emph{cap-limited} and by \eqref{eq:EM-cap1}
\begin{equation}
\; \mathbb E[M] \;\sim\; \frac{N}{w_0\log N}\;  ,
  \label{eq:area-cap-scaling}
\end{equation}
yielding quasi-linear sampling per pass, $\mathcal O(N\log N)$.

\paragraph{Weak long-range power law ($\alpha>1$; resummable tail).}
The mean-field error bound gives a \emph{linear} hazard ($\gamma=1$) with a window–controlled prefactor
\begin{equation}
  p_0 \sim \frac{C}{w^{\alpha}},\qquad \alpha>1.
  \label{eq:plr-p0}
\end{equation}
The associated hazard length is $\ell \sim p_0^{-1/2}\sim w^{\alpha/2}$. Let the window scale as
$w(N)=w_0 N^{\beta}$, and recall the cap size
\begin{equation}
  K \sim \frac{N}{w} \sim N^{1-\beta}.
  \label{eq:cap-K2}
\end{equation}

 When $K\gg \ell$, the process freezes from the hazard before hitting the cap, and
\begin{equation}
  \mathbb{E}[M] \;\sim\; \ell \;\sim\; p_0^{-1/2} 
  \;\sim\; N^{\frac{\alpha\beta}{2}}.
  \label{eq:EM-hazard}
\end{equation}
This \textbf{increases} monotonically with $\beta$ (a larger window suppresses the hazard, allowing more accepted flips). When $K\ll \ell$, the chain almost never freezes early and runs to the cap,
\begin{equation}
  \mathbb{E}[M] \;\sim\; K \;\sim\; \frac{N}{w}
  \;\sim\; N^{1-\beta}.
  \label{eq:EM-cap}
\end{equation}
This \textbf{decreases} monotonically with $\beta$ (a larger window lowers the cap).

Since the hazard-limited value grows with $\beta$ while the cap-limited value shrinks with $\beta$, the optimum is attained precisely at the \emph{crossover} where the two branches match.
\begin{equation}
  N^{1-\beta} \;\sim\; N^{\frac{\alpha}{2}\beta}
  \qquad\Longrightarrow\qquad
 \,\beta_c=\frac{2}{2+\alpha}\,  .
  \label{eq:crossover2}
\end{equation}
At this $\beta_c$, the two expressions \eqref{eq:EM-hazard} and \eqref{eq:EM-cap} coincide and yield
\begin{equation}
\,\mathbb{E}[M]_{\mathrm{opt}}
  \;\sim\; N^{1-\beta_c}
  \;=\; N^{\frac{\alpha}{2+\alpha}}\, .
  \label{eq:plr-opt}
\end{equation}

\begin{figure}[H]
    \centering
    \includegraphics[width=1.0\linewidth]{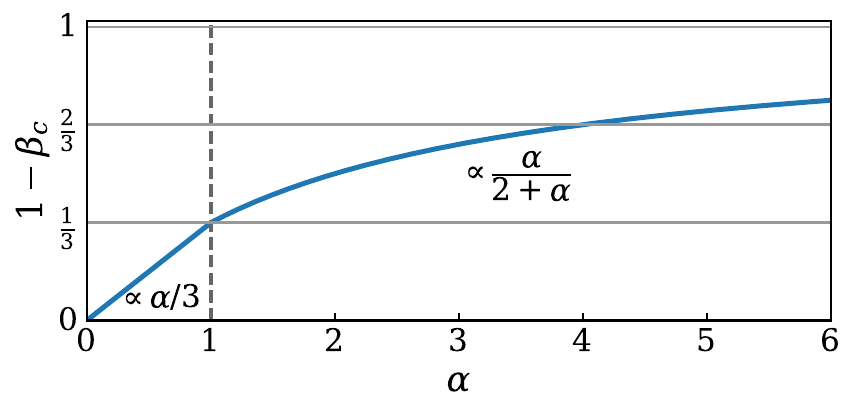}
    \caption{\textbf{Optimal window exponent for the screened typewriter sampler.}
We plot \(1-\beta_c\) versus the interaction–range exponent \(\alpha\), where the optimal spacing scales as \(w(N)\sim N^{\beta_c}\) and the expected number of accepted flips per sweep scales as \(\mathbb{E}[M]\sim N^{1-\beta_c}\). The dashed line marks the crossover at \(\alpha=1\). For non–resummable power–law tails \(0<\alpha<1\) one finds \(\beta_c=(3-\alpha)/3\) (hence \(1-\beta_c=\alpha/3\), left annotation). For resummable tails \(\alpha>1\) one obtains \(\beta_c=2/(2+\alpha)\) (so \(1-\beta_c=\alpha/(2+\alpha)\), right annotation). These exponents imply the per–flip runtime scalings quoted in the text, \(T_{\text{per flip}}=\mathcal{O}(N^{\beta_c}\log N)\).}
    \label{fig:exponentAlpha}
\end{figure}

\paragraph{Non-resummable power law ($0<\alpha<1$; superlinear hazard).}
Pairwise-dominated cross terms yield a \emph{superlinear} hazard
\begin{equation}
  p_k = p_0\,k^\gamma,\qquad \gamma=2-\alpha>1,\qquad
  p_0 \sim \frac{\tilde C}{w^{\alpha}},
  \label{eq:nrpl-hazard}
\end{equation}
so that $\gamma{+}1=3-\alpha$ and the hazard length is
$\ell \sim p_0^{-1/(3-\alpha)} \sim w^{\alpha/(3-\alpha)}$.
Let $w(N)=w_0 N^\beta$ and recall the cap \eqref{eq:cap-K2}, $K\sim N^{1-\beta}$. If $K\gg \ell$, the process freezes before the cap, and by analogy with \eqref{eq:EM-hazard} one has
\begin{equation}
  \mathbb{E}[M] \;\sim\; \ell
  \;\sim\; p_0^{-1/(3-\alpha)} \;\sim\; w^{\alpha/(3-\alpha)}
  \;\sim\; N^{\frac{\alpha\beta}{3-\alpha}}.
  \label{eq:nrpl-EM-hazard}
\end{equation}
This \textbf{increases} monotonically with $\beta$ (a larger window suppresses the hazard more strongly).

 If $K\ll \ell$, the chain essentially always reaches the cap, just as in \eqref{eq:EM-cap}.
\begin{equation}
  \mathbb{E}[M] \;\sim\; K \;\sim\; \frac{N}{w}
  \;\sim\; N^{1-\beta}.
  \label{eq:nrpl-EM-cap}
\end{equation}
This \textbf{decreases} monotonically with $\beta$ (a larger window lowers the cap).

Because the hazard-limited value rises with $\beta$ while the cap-limited value falls with $\beta$, the optimum is achieved at the \emph{crossover}, exactly as argued in \eqref{eq:crossover2}. Equating the two scalings gives
\begin{equation}
  \,\beta_c = \frac{3-\alpha}{3}\, .
  \label{eq:nrpl-crossover}
\end{equation}
At $\beta=\beta_c$, the branches \eqref{eq:nrpl-EM-hazard} and \eqref{eq:nrpl-EM-cap} coincide, yielding
\begin{equation}
\,\mathbb{E}[M]_{\mathrm{opt}} \;\sim\; N^{\,1-\beta_c}
  \;=\; N^{\alpha/3}\,  .
  \label{eq:nrpl-opt}
\end{equation}

 From the weak long-range regime (linear hazard) we had
$\beta_c=\frac{2}{2+\alpha}$; at $\alpha=1$ this gives $\beta_c=2/3$. Our non-resummable result
\eqref{eq:nrpl-crossover} also yields $\beta_c=1-\alpha/3=2/3$ at $\alpha=1$. Thus the optimal window exponent matches smoothly across the boundary $\alpha=1$.

 Plugging $\alpha=0$ into \eqref{eq:nrpl-crossover} gives
$\beta_c=1$, i.e.\ $w\propto N$ so $K\sim N/w$ is constant. Correspondingly,
\eqref{eq:nrpl-opt} gives $\mathbb{E}[M]_{\mathrm{opt}}\sim N^{0}$: there is \emph{no asymptotic sampling advantage} in the fully connected ($\alpha=0$) case, although \emph{constant-factor} speedups (from better constants in $p_0$ or reduced autocorrelation) can still be realized in practice.

\section{Practical implementation of the screened typewriter sampler}

Implementing the screened typewriter sampler on GPUs requires modifying Algorithm 2 to maximize parallelism. We process proposals for an entire batch of chains simultaneously. Because the Independent Scattering Approximation (ISA) renders defects non-interacting, we parallelize computationally intensive steps—such as link tensor construction and ABACUS recurrence—both across chains and across the proposals within a single chain. With amplitudes precomputed, a lightweight loop handles acceptance. To accommodate the limited control flow of GPUs, we do not halt execution when a chain freezes; instead, we use masking to prevent updates to the output buffer, preserving batch synchronization. Optimizing this masking scheme could yield further speedups.

The screened typewriter sampler requires exact error bounds to remain exact. These bounds can be computed by evaluating Eq.~\eqref{eq:deltaz-pre} and Eq.~\eqref{eq:key-result}, which guarantees that the screened typewriter sampler is provably unbiased. Additionally, we provide a pragmatic strategy that is easier to implement than evaluating Eq.~\eqref{eq:key-result} at every step. We allocate a buffer $\epsilon(k)$ indexed by the number of accepted flips $k$. We initialize $\epsilon(k)$ to a conservative value (e.g., 0.2) to ensure it exceeds the actual error, thereby forcing the chain to freeze prematurely. During these freeze events, we obtain the empirical error $\tilde{\epsilon}$ at no additional computational overhead by comparing the exact and approximate acceptance ratios. If $\epsilon(k)$ retains its initial default value, we overwrite it with $\tilde{\epsilon}$; otherwise, we update $\epsilon(k) \leftarrow \max\{ \epsilon(k),\gamma \tilde{\epsilon}\}$ with $\gamma>1$. This heuristic converges rapidly, self-corrects as the wavefunction evolves, and avoids calibration parameters. We note, however, that it is still a heuristic and can underestimate $\epsilon(k)$ early in training, which may introduce bias. For this reason, we discard the first few sweeps at the beginning of each training run. Empirically, after this burn-in, the error bounds are accurate and the sampler satisfies the Metropolis--Hastings rule exactly. All final observables we report (ground-state energy, observables, V-score, etc.) are obtained with a provably exact sampler.

\label{app:throughput}

\begin{figure}
    \centering
    \includegraphics[width=1.0\linewidth]{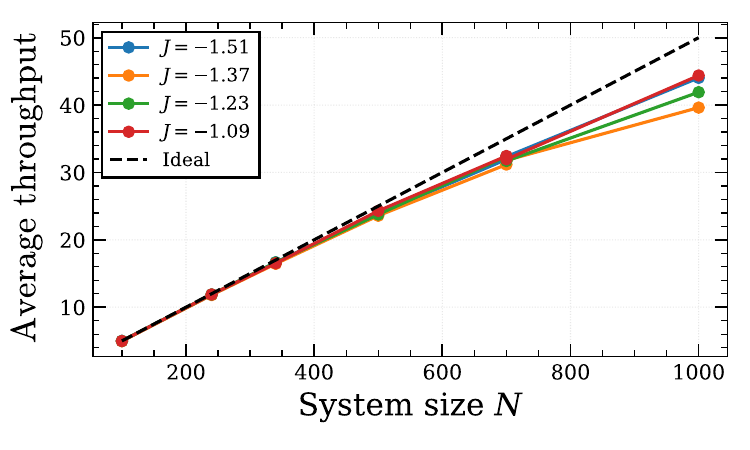}
    \caption{Scaling of average sampler throughput with system size $N$ for interaction range $\alpha=1.5$. The dashed black line indicates the theoretical maximum throughput $N/2W$ in the ideal limit of zero screening. Colored symbols represent different coupling strengths $J$. The throughput scales approximately linearly, demonstrating efficient amortization of the $\mathcal{O}(N \log N)$ update cost. Deviations from the ideal line at large $N$ signal the onset of premature freezing, an effect that is most pronounced near the critical point ($J_c \approx -1.4$) due to enhanced long-range correlations.}
    \label{fig:throughput}
\end{figure}

Figure~S\ref{fig:throughput} plots the \emph{average sampler throughput} of the screened typewriter scheme. Throughput is defined as the average number of proposals per chain between cache refreshes. Since link-tensor construction is the expensive $\mathcal{O}(N\log N)$ step, this metric quantifies how well the scheme amortizes that cost.

A dashed line marks the ideal case, where screening never triggers during a pass. With block spacing $W$, a pass contains
\begin{equation}
M_{\rm ideal}(N) \;=\; \frac{N}{2W}.
\end{equation}
proposals per chain, yielding the linear reference.

We examine models trained at $\alpha=1.5$. This setting poses the greatest challenge: long-range correlations reduce the accuracy of the independent-scattering approximation (ISA). Nevertheless, throughput scales linearly with $N$, tracking the ideal line. The curves diverge only at the largest system size ($N=1000$), signaling \emph{premature freezing}. Here, screening triggers early, forcing a cache refresh before the cap is complete. Even in this regime, the sampler retains 80--90\% of the ideal throughput.

Consequently, we keep $W$ constant. The linear scaling demonstrates that fixed spacing efficiently amortizes cache costs. The throughput drop at $N=1000$ does not justify increasing $W$; adjusting the interval width becomes necessary only for $N > 1000$.

Finally, the figure reveals the sampler's sensitivity to ground-state correlations. Throughput peaks in the ferromagnetic and paramagnetic regimes but falls near the critical point, $J_c \approx -1.4$. There, long-range correlations increase ISA interference, forcing proposals into the ambiguous window and triggering early cache refreshes.

\section{Efficient construction of link tensors}\label{app:linkTensorConstruction}
\emph{Final-stage links $T^{(l)}$.}
The final-stage link tensors $T^{(l)}$ only have rank $d$. They can therefore be computed using $d$ reverse matrix vector products. In all machine learning libraries this can be computed highly efficiently via $d$ backwards passes. 

\emph{Two-layer links $L^{(l,l-1)}$.}
The two-layer link tensors $L^{(l, l-1)}_{j, \Delta j}$ can be efficiently computed using convolutions. Specifically, we have 

\begin{align}
    L^{(l, l-1)}_{j, \Delta j} & = \sum_k \boldsymbol G_{j + \Delta j-k}^{(l)}  D_k^{(l-1)} \boldsymbol G_{k-j}^{(l-1)}  \ . 
\end{align}

Let's define the composite Green's function ${\tilde G}_{j-k, \Delta j}^{(l)} = \boldsymbol G_{j -k + \Delta j}\otimes \boldsymbol G_{-(j -k)}$. This is another convolution kernel meaning we can efficiently construct the link tensor via the matrix valued FFT 

\begin{align}
    L^{(l, l-1)}_{j, \Delta j} = \mathcal F^{-1} ( \mathcal F ( \tilde G_{\Delta j}^{(l)}) \odot  \mathcal F (D^{(l-1)}))\vert_j \ . 
\end{align}

\noindent We need this link tensor for all $j \in \{1, \dots, N\}$ but for $\Delta j \in \{-w, \dots, w+1\}$. This means the link tensor can be constructed in $\mathcal O(d^2 W N \log N) \subset \mathcal O(N\log N)$.

\emph{Multi-layer links $|l-m|\ge 3$.}
Direct nested FFTs scale as badly for more than two layers. We can instead use randomized SVDs to obtain the link tensors
\(
L^{(l,m)} \approx U^{(l,m)}\,\tilde\Sigma^{(l,m)}_k\,V^{(l,m)}
\)
 using $k\ll W d$ random probes (linear-time in $N$ for fixed $k$) and a small $\mathcal O(k^3)$ factorization on slice-shaped matrices. In the worst case $k=W d$ this incurs a total cost of $\mathcal O(W d^2 N\log N+ W d^3 N) \subset \mathcal O(N\log N)$ matrix vector multiplication cost and $\mathcal O(W^3 d^3)$ build cost. Note using randomized SVDs the error can be made arbitrarily small. Nevertheless, in this paper we empirically only need two layers to obtain accurate results. Therefore, we do not require this construction, meaning the ABACUS algorithm returns provably exact results without controlling for error.

\section{Exact algorithm to construct link tensors for arbitrary network depth}
\label{sec:hodlr_greens}

In our benchmarks, two layers of \cobalt\  suffice for high performance, allowing us to construct all link tensors efficiently using FFTs. However, because nested FFT contraction scales poorly with depth, this strategy fails for deeper networks. To demonstrate that ABACUS is mathematically exact and asymptotically efficient, we provide an advanced algorithm for constructing link tensors, even though our experiments do not require it.

We leverage additional structure in the long-range token mixer to compute link tensors efficiently. Consistent with our assumption that long-wavelength propagation is controlled by only a few modes, we model each Green's-function mixer $G^{(\ell)}\in\mathbb{R}^{N\times N}$ as a hierarchical off-diagonal low-rank (HODLR) matrix. HODLR matrix multplication is highly efficient allowing us to build the link tensors at reasonable cost.   We show that SSM impulse-response kernels (finite sums of exponentials) are exactly HODLR.

\subsection{Definition of HODLR matrices}
\label{sec:hodlr_def}

Fix an ordering of token indices $\{1,\dots,N\}$ along a 1D chain (or along a chosen space-filling order in higher
dimensions), so that contiguous index sets correspond to local spatial neighborhoods.
For simplicity assume $N=2^p$ and define a recursive binary partition of the index interval into contiguous halves.
At the top level,
\begin{equation}
I_1 := \{1,\dots,N/2\},\qquad I_2 := \{N/2+1,\dots,N\}.
\end{equation}

\begin{definition}[HODLR matrix]
\label{def:hodlr}
A matrix $A\in\mathbb{R}^{N\times N}$ is \textbf{HODLR of rank $\le r$} (with leaf size $n_0$) if it can be represented as
\begin{align}
\label{eq:hodlr_level0}
A \;=\;
\begin{pmatrix}
A_{I_1,I_1} & A_{I_1,I_2}\\[2pt]
A_{I_2,I_1} & A_{I_2,I_2}
\end{pmatrix},
A_{I_1,I_2}=U_{12}V_{12}^\top,\quad
A_{I_2,I_1}=U_{21}V_{21}^\top,
\end{align}
where $\mathrm{rank}(A_{I_1,I_2}),\mathrm{rank}(A_{I_2,I_1})\le r$, and where the \emph{diagonal blocks}
$A_{I_1,I_1}$ and $A_{I_2,I_2}$ are again HODLR of rank $\le r$ under the next binary split of $I_1$ and $I_2$,
recursively, until blocks reach size $\le n_0$ (stored densely).
Importantly, only the \emph{diagonal blocks are refined recursively}; off-diagonal blocks are stored as low-rank factors
at the level where they appear.
\end{definition}

\subsection{SSM kernels are exactly HODLR}
\label{sec:ssm_exact_hodlr}

Consider a 1D translation-invariant Green's function (Toeplitz kernel)
\begin{equation}
\label{eq:toeplitz_kernel}
G_{ij} \;=\; g(|i-j|), \qquad i,j\in\{1,\dots,N\}.
\end{equation}
We focus on SSM impulse responses, which in real space are finite sums of (possibly damped/oscillatory) modes.
For clarity we first treat real exponentials, then state the extension.

\begin{proposition}[Sum of exponentials $\Rightarrow$ exact low-rank off-diagonal blocks]
\label{prop:ssm_exact_hodlr}
Let
\begin{equation}
\label{eq:ssm_sum_exp}
g(t)=\sum_{p=1}^{s}\alpha_p \lambda_p^{t},
\qquad |\lambda_p|<1.
\end{equation}
Let $I=\{a,\dots,b\}$ and $J=\{c,\dots,d\}$ be two disjoint contiguous index intervals with $b<c$
(i.e.\ $I$ lies strictly to the left of $J$). Then the off-diagonal block $G_{I,J}$ has rank at most $s$ and factors as
\begin{equation}
\label{eq:ssm_block_factor}
G_{I,J}
\;=\;
\sum_{p=1}^{s}
u^{(p)}_I \bigl(v^{(p)}_J\bigr)^\top,
\qquad
u^{(p)}_I(i)=\lambda_p^{-i},\quad
v^{(p)}_J(j)=\alpha_p\,\lambda_p^{j}.
\end{equation}
The analogous statement holds for blocks strictly below the diagonal. In particular, for any recursive binary partition
into contiguous halves (Definition~\ref{def:hodlr}), every off-diagonal block is exactly rank $\le s$, hence $G$ is an exact
HODLR matrix with rank bound $r\le s$.
\end{proposition}

\begin{proof}
If $i\in I$ and $j\in J$ with $i<j$, then $|i-j|=j-i$ and
\begin{equation}
G_{ij}=\sum_{p=1}^{s}\alpha_p \lambda_p^{j-i}
=\sum_{p=1}^{s} \bigl(\lambda_p^{-i}\bigr)\bigl(\alpha_p \lambda_p^{j}\bigr),
\end{equation}
which is precisely the separable expansion in Eq.~\eqref{eq:ssm_block_factor}.
\end{proof}

\begin{remark}[Damped oscillations / complex poles]
Many SSM kernels are sums of damped oscillatory modes corresponding to complex poles inside the unit disk.
Writing such modes as complex exponentials (or equivalently as real $\exp(-t/\xi)\cos(\omega t)$ and
$\exp(-t/\xi)\sin(\omega t)$ pairs) preserves the same separability argument; the exact HODLR rank bound is
proportional to the number of modes (counting complex-conjugate pairs appropriately).
\end{remark}

\subsection{Link tensor construction}

Let the network have depth $L$ with layers indexed by $\ell=1,\dots,L$.
Let $D^{(\ell)}(\boldsymbol \sigma_0)$ denote the frozen-background diagonal (local) vertex at layer $\ell$
(see main text), and let $G^{(\ell)}$ be the corresponding Green's-function token mixer.
Define the composite layer operator
\begin{align}
\label{eq:B_def}
B^{(\ell)} &:= D^{(\ell)}(\boldsymbol \sigma_0)\,G^{(\ell)}.
\end{align}
\noindent
For integers $1\le a\le b\le L$, define the consecutive product
\begin{align}
\label{eq:X_def}
X_{a,b} &:= B^{(a)}B^{(a+1)}\cdots B^{(b)}, \qquad X_{a,a}=B^{(a)},
\end{align}
\noindent
and use the convention
\begin{align}
\label{eq:X_empty}
X_{a,b} &:= I \qquad (a>b).
\end{align}
\noindent
Let $P_{S,j}$ be the projector onto the local patch $S_j$ (width $W$) centered at site $j$.
For any pair of layers $\ell>m$, we cache the local interlayer link tensor as the patch-restricted block
\begin{align}
\label{eq:L_def}
\mathcal{L}^{(\ell,m)}_j
&:= P_{S,j}\,\Bigl(G^{(\ell)}\,X_{m+1,\ell-1}\Bigr)\,P_{S,j}^{\top}
\in \mathbb{R}^{dW\times dW}.
\end{align}
\noindent
For $\ell=m+1$, the segment is empty and Eq.~\eqref{eq:X_empty} gives $X_{m+1,\ell-1}=I$.

\subsection{Parallel construction via scans}

The previous section outlined how to construct the link tensors using HODLR multiplication. This section targets construction of multiple link tensors in parallel on GPUs. Naivly, one would need $\mathcal O(L^3)$ HODLR multiplications to form all link tensors. However, since HOLDR multipication is associative, one can reduce the cost to $O(L^2 \log L)$ using a parallel prefix (Blelloch) scan \cite{blelloch_synthesis_1990}. 

Fix an endpoint $b\in\{1,\dots,L\}$ and consider the reversed list
$\bigl(B^{(b)},B^{(b-1)},\dots,B^{(1)}\bigr)$.
A parallel prefix-scan over this list yields all prefix products
\begin{align}
\label{eq:prefix_def}
P^{(b)}_t
&:= B^{(b)} \odot B^{(b-1)} \odot \cdots \odot B^{(b-t+1)},
\qquad t=1,\dots,b.
\end{align}
\noindent
These prefixes correspond exactly to consecutive products ending at $b$ via the index map
\begin{align}
\label{eq:X_from_prefix}
X_{a,b} &= P^{(b)}_{\,b-a+1}.
\end{align}
\noindent
Thus, one scan computes the entire family $\{X_{a,b}\}_{a=1}^{b}$. Running the same scan independently for
different endpoints is precisely “parallel scans starting at different nodes” of the layer chain.

To construct all link tensors, for each target layer $\ell$ we run a scan with endpoint $b=\ell-1$ to obtain
all segments $\{X_{m+1,\ell-1}\}_{m=0}^{\ell-1}$ using Eq.~\eqref{eq:X_from_prefix}.
For each $m$, we then form $G^{(\ell)}X_{m+1,\ell-1}$ (HODLR multiply+recompress) and extract all local blocks
$\mathcal{L}^{(\ell,m)}_j$ using Eq.~\eqref{eq:L_def}; all patch indices $j$ are independent and hence parallel.

\subsection{Cost}

Let one HODLR multiply+recompress cost
\begin{align}
\label{eq:Tmm}
T_{\mathrm{mm}}(N,r)
&= \Theta\!\bigl(N r^2 \log(N/n_0)\bigr).
\end{align}
\noindent
A scan of length $b$ uses $\Theta(b)$ HODLR multiplications (linear work) and has $\Theta(\log b)$ scan stages
(logarithmic depth) \cite{blelloch_synthesis_1990}. Summing the work over endpoints $b=1,\dots,L$ gives $\Theta(L^2)$
HODLR multiplications, hence
\begin{align}
\label{eq:work_total}
T_{\mathrm{work}}
&= \Theta(L^2)\,T_{\mathrm{mm}}(N,r)
= \Theta\!\bigl(L^2 N r^2 \log(N/n_0)\bigr).
\end{align}
\noindent
Running all endpoint scans in parallel yields overall depth
\begin{align}
\label{eq:depth_total}
T_{\mathrm{depth}}
&= \Theta(\log L)\,T_{\mathrm{mm}}(N,r).
\end{align}
\noindent
Forming $B^{(\ell)}$ in Eq.~\eqref{eq:B_def} by diagonal scaling is lower order than Eq.~\eqref{eq:work_total},
and patch-block extraction in Eq.~\eqref{eq:L_def} is embarrassingly parallel over $(\ell,m,j)$.

\noindent
This scan-based construction therefore computes all cached link tensors with $\Theta(\log L)$ depth and
log-linear HODLR costs in the system size $N$. This shows link tensors for ABACUS can be efficiently constructed independent of network depth.

\section{DysonNet in two dimensions}
\label{app:ABACUS2D}

We briefly extend the \cobalt\  block to a 2D lattice with $N=L_xL_y$ sites. We index lattice sites by $j=(j_x,j_y)$ and write translation-invariant kernels as functions of the displacement $j-k$. The \cobalt\  block remains
\begin{equation}
\Phi(\boldsymbol \sigma)\,h \;=\; D(\boldsymbol \sigma)\,G\,h,
\end{equation}
with the same residual stacking and mean-pooling readout $P_M$ as in the main text.

The 2D token mixer $G$ is a translationally invariant convolution operator (Toeplitz/circulant under periodic boundaries),
\begin{equation}
(Gh)_{j,c} \;=\; \sum_{r} G_{j-r,c}\,h_{r,c},
\end{equation}
which is implemented with 2D FFTs channel-wise.
\begin{equation}
Gh \;=\; \mathcal{F}_2^{-1}\!\big(\mathcal{F}_2(G)\odot \mathcal{F}_2(h)\big),
\end{equation}
where $\mathcal{F}_2$ denotes the discrete 2D Fourier transform over the $L_x\times L_y$ grid and $\odot$ is pointwise multiplication.

The local nonlinearity $D(\boldsymbol \sigma)$ acts position-wise on each embedding vector, but may depend on a bounded \emph{patch} around $j$. Concretely, for a fixed patch radius $w$ define the square patch
\begin{equation}
S_j \;:=\;\{\,k:\ \|k-j\|_\infty \le w\,\}, \qquad W := |S_j|=(2w+1)^2,
\end{equation}
and parameterize the local vertex as
\begin{equation}
(D(\boldsymbol \sigma)h)_{j} \;=\; D\!\big(\boldsymbol \sigma|_{S_{j}}\big)\,h_j,
\end{equation}
i.e.\ $D(\cdot)$ ``attends'' to the $W$ tokens in the patch (via a small CNN/MLP over the patch) and outputs a $d\times d$ matrix multiplying the center token $h_j$.

For a single spin flip at position $j$ relative to a reference configuration $\boldsymbol \sigma_0$, we decompose each layer's vertex as
\begin{equation}
D^{(l)}(\boldsymbol \sigma) \;=\; D^{(l)}_0 + \Delta D^{(l)}, \qquad D^{(l)}_0 := D^{(l)}(\boldsymbol \sigma_0),
\end{equation}
where $\Delta D^{(l)}$ is nonzero only on a bounded region: in 2D it is supported on (a small union of) patches around the flip, rather than on an interval. Let $P_{S,j}$ denote the projector that restricts activations to the patch $S_j$ (so $P_{S,j}h\in\mathbb{R}^{W\times d}$). Then the link tensors are defined exactly as before, with $P_{S,j}$ now projecting onto a patch.
\begin{align}
T^{(l)} \;&=\; P_M\,G^f\!\left(\sum_{M=l}^{L}\prod_{n=l+1}^{M} D^{(n)}_0G^{(n)}\right)P_{S,j}^{\top},\\
L^{(l,m)} \;&=\; P_{S,j}\,G^{(l)}\!\left(\prod_{n=m+1}^{l-1} D^{(n)}_0G^{(n)}\right)P_{S,j}^{\top}, \quad 1\le m\le l-2.
\end{align}
Defining patch-restricted activations $\tilde h^{(l)} = P_{S,j}h^{(l)}\in\mathbb{R}^{W\times d}$, $\tilde h^{(l)}_0=P_{S,j}h^{(l)}_0$, and $\tilde G^{(l)} := P_{S,j}G^{(l)}P_{S,j}^{\top}$, the \cobalt/ABACUS recurrence is unchanged except for interpreting the restriction as a patch.
\begin{align}
\tilde h^{(l)}
&=\Delta D^{(l)}
\left(
\tilde h^{(l)}_0
+ \tilde G^{(l)}\tilde h^{(l-1)}
+ \sum_{m=1}^{l-2} L^{(l,m)}\tilde h^{(m)}
\right),\\
\qquad
\boldsymbol\Omega
&= A\!\left(P_M\tilde h^{(L)} + \sum_{l=1}^{L} T^{(l)}\tilde h^{(l)}\right).
\end{align}
Here $\tilde G^{(l)}\tilde h^{(l-1)}$ denotes propagating a patch-supported input through the convolution and then restricting back to the patch; operationally this is a contraction of the local $(W\times W)$ block of the Toeplitz operator (or an equivalent SSM restriction), exactly as in 1D.

Complexity statements carry over verbatim upon replacing the slice width $W=2w+1$ by the patch size $W=(2w+1)^2$. Given precomputed links, the local-update runtime remains independent of $N$ and scales as  $O(L^2 W^2 d^2)$,
Importantly, the link tensors depend on the \emph{number of tokens in the restricted region} ($W$), not on whether that region is a 1D interval or a 2D area. For \cobalt, link construction remains log-linear in $N$ because all required background propagations reduce to 2D convolutions, which are computed with 2D FFTs (and multi-layer links can be handled by the same low-rank probing strategy, with each probe dominated by $O(N\log N)$ FFT-based matvecs for fixed $W,d$.

\section{Details on the implementation of the DysonNet block}
\label{app:cobalt-details}
Here we provide additional details on the implementation of S4 and the CNN.

Here we summarize the S4 parameterization from Ref.~\cite{gu_efficiently_2022} in notation consistent with the main-text definitions in Eqs.~\eqref{eq:ssm_discretized} and \eqref{eq:resolvent}. We start from the continuous-time state equation
\begin{align}
    \partial_x \boldsymbol y(x) = \bar A\,\boldsymbol y(x) + \bar B\,\boldsymbol h(x).
\end{align}

\noindent We discretize it with the bilinear transform
\begin{align}
     A &= \Bigl(I - \frac \Delta2 \bar A\Bigr)^{-1}\Bigl(I+\frac \Delta2 \bar A\Bigr), \\
     B &= \Bigl(I - \frac{\Delta} 2\bar A\Bigr)^{-1}\Delta \bar B. \label{eq:discretization}
\end{align}

\noindent where $\Delta$ is the step size. For $z_k=e^{i2\pi k/N}$, Eq.~\eqref{eq:resolvent} becomes
\begin{align}
    G_c(z_k)
    &= C_c\,(I-z_k A_c)^{-1} B_c \\
    &= C_c\,c(z_k)\,\bigl(g(z_k)I-\Delta \bar A_c\bigr)^{-1}\bar B_c .
\end{align}
with
\begin{align}
    c(z)=\frac{2\Delta}{1+z},
    \qquad
    g(z)=\frac{2(1-z)}{1+z}.
\end{align}
For diagonal $\bar A_c=\Lambda_c$, define the scalar Cauchy-kernel contraction
\begin{align}
    \widehat G_{0,c}(z;u,v)=\sum_{i=1}^{s}\frac{v_i u_i}{g(z)-\Delta \Lambda_{c,i}} .
\end{align}

\noindent For the S4 parameterization $\bar A_c=\Lambda_c+\boldsymbol p_c\boldsymbol q_c^{\top}$ (matching the main text), Woodbury gives~\cite{gu_efficiently_2022}
\begin{align}
\begin{split}
    G_c(z) &= c(z)\Big( \widehat G_{0,c}(z;\bar B_c,C_c) \\
    &\quad - \widehat G_{0,c}(z;\boldsymbol p_c,C_c)\,[1+\widehat G_{0,c}(z;\boldsymbol p_c,\boldsymbol q_c^{*})]^{-1}\widehat G_{0,c}(z;\bar B_c,\boldsymbol q_c^{*})\Big). \label{eq:greenWoodbury}
\end{split}
\end{align}

\noindent This is the Green's function practically implemented in our network.

For parameterizing the local vertex $D(\boldsymbol \sigma)$ we use a short-ranged CNN kernel, as this enforces a finite receptive field $D_j(\boldsymbol \sigma)=D_j(\sigma_{j-w},\dots,\sigma_{j+w})$. In our model we create a second data stream $\boldsymbol{\phi}^{(l)}$ (see Fig.~\ref{fig:COBALTPractical}) with vectors $\boldsymbol \phi_j^{(l)}\in \mathbb R^d$ for sites $j=1,\dots,N$. We compute $\boldsymbol \phi^{(1)}=E_\phi\boldsymbol \sigma$ with embedding matrix $E_\phi\in \mathbb R^{d\times 1}$. Each layer applies (i) a site-wise dense map, (ii) a short-range depth-wise convolution, and (iii) an elementwise nonlinearity $g$ (we use $\mathrm{SiLU}$).

\begin{align}
\boldsymbol {\tilde \phi}_j^{(l)} &= W_D^{(l)}\boldsymbol \phi_j^{(l)} + \boldsymbol w^{(l)}, \\
\boldsymbol {\phi}_j^{(l+1)} &= g\!\bigg( \sum_{r=-w}^w K_r^{(l)} \odot \boldsymbol {\tilde \phi}_{j+r}^{(l)} + \boldsymbol k^{(l)} \bigg).
\end{align}

\noindent 
Here, 
$W_D^{(l)}\in\mathbb{R}^{d\times d}$ and $\boldsymbol w^{(l)}\in\mathbb{R}^{d}$ are the site-wise
(dense) weight matrix and bias at layer $l$.
The convolution has half-width $w$ with offsets $r\in[-w,w]$; indices $j+r$ are taken modulo $N$
(circular padding). $K_r^{(l)}\in\mathbb{R}^{d}$ is the kernel at offset $r$ and
$\boldsymbol k^{(l)}\in\mathbb{R}^{d}$ its bias.
The nonlinearity is applied element-wise and chosen as $\mathrm{SiLU}$:
$g(x)=\mathrm{SiLU}(x)=x/(1+e^{-x})$. \\ 

To match Eq.~\eqref{eq:Dpractical}, we define
\begin{align}
    \chi(\boldsymbol\phi^{(l-1)}_j) &= g\!\big(A_D^{(l)} \boldsymbol\phi^{(l-1)}_j + \boldsymbol b^{(l)}\big), \\
    D_j(\boldsymbol \sigma)  &= U^{(l)} \,\mathrm{diag}\!\big(\chi(\boldsymbol\phi^{(l-1)}_j)\big)\, V^{(l)} .
\end{align}
Here $A_D^{(l)}\!\in\!\mathbb{R}^{d\times d}$ and $\boldsymbol b^{(l)}\!\in\!\mathbb{R}^{d}$ are trainable dense weights, and $g$ is an elementwise nonlinearity (we use $\mathrm{SiLU}$). The vector $\chi(\boldsymbol\phi^{(l-1)}_j)\!\in\!\mathbb{R}^{d}$ parametrizes the configuration-dependent spectrum, while $U^{(l)},V^{(l)}\!\in\!\mathbb{R}^{d\times d}$ are layer-specific, configuration-independent bases (fixed left/right “eigenvectors”). Thus $D_j(\boldsymbol \sigma)\!\in\!\mathbb{R}^{d\times d}$ has fixed eigenvectors but spin-configuration–dependent eigenvalues through $\boldsymbol\phi^{(l-1)}_j$. As in the main text, we also inject the $\boldsymbol\phi^{(l)}$ stream additively into the mixer stream before applying the SSM.
In this work we take the parameters to be real. We note as this is inside the network this does not limit expressivity to real valued wavefunctions, as a complex vector can be embedded in a real-valued vector space. The final layer determines the sign structure of the network \footnote{If a symmetric parameterization is desired, one may set $V^{(l)}\!=\!U^{(l)\top}$.}.

\section{Small system validation TFIM}

We evaluate the neural quantum state (NQS) for the long-ranged transverse-field Ising model at $N=20$ on a 2D parameter grid—vertical axis $J\in[-5,5]$, horizontal axis $\alpha\in[0,6]$; see  Fig.~\ref{fig:infidelity}. At each grid point we report the overlap error (infidelity) with the exact ground state. 
The fidelity is $F=|\langle\Psi_0|\Psi_{\mathrm{NQS}}\rangle|^2$; in the presence of a (near-)degenerate ground-state manifold we compute the overlap with the ground-state subspace, $F=\sum_j|\langle\Psi_{0,j}|\Psi_{\mathrm{NQS}}\rangle|^2$, treating states as degenerate when their energy splitting is $<10^{-6}$. The right panel of Fig.~\ref{fig:infidelity} shows the ground-state energy error, $\Delta E=|E_0-E_{\mathrm{NQS}}|$. For each grid point we perform three independent training runs and report the median; no training instabilities were observed.

Across the grid, the median infidelity is $2.99\times10^{-4}$: most points sit at $\lesssim10^{-3}$, with a few localized pockets reaching $\sim10^{-2}$ near short-ranged (large-$\alpha$) critical regions where optimization is harder---still small in absolute terms. Ground-state energies are uniformly accurate: typical $\Delta E$ is $2$--$5\times10^{-4}$, spanning $8.83\times10^{-5}$ to $9.54\times10^{-4}$. Overall, energies are robust across the phase diagram, while fidelity remains tight almost everywhere and softens only mildly near criticality.

\begin{figure}
    \centering
    \includegraphics[width=\linewidth]{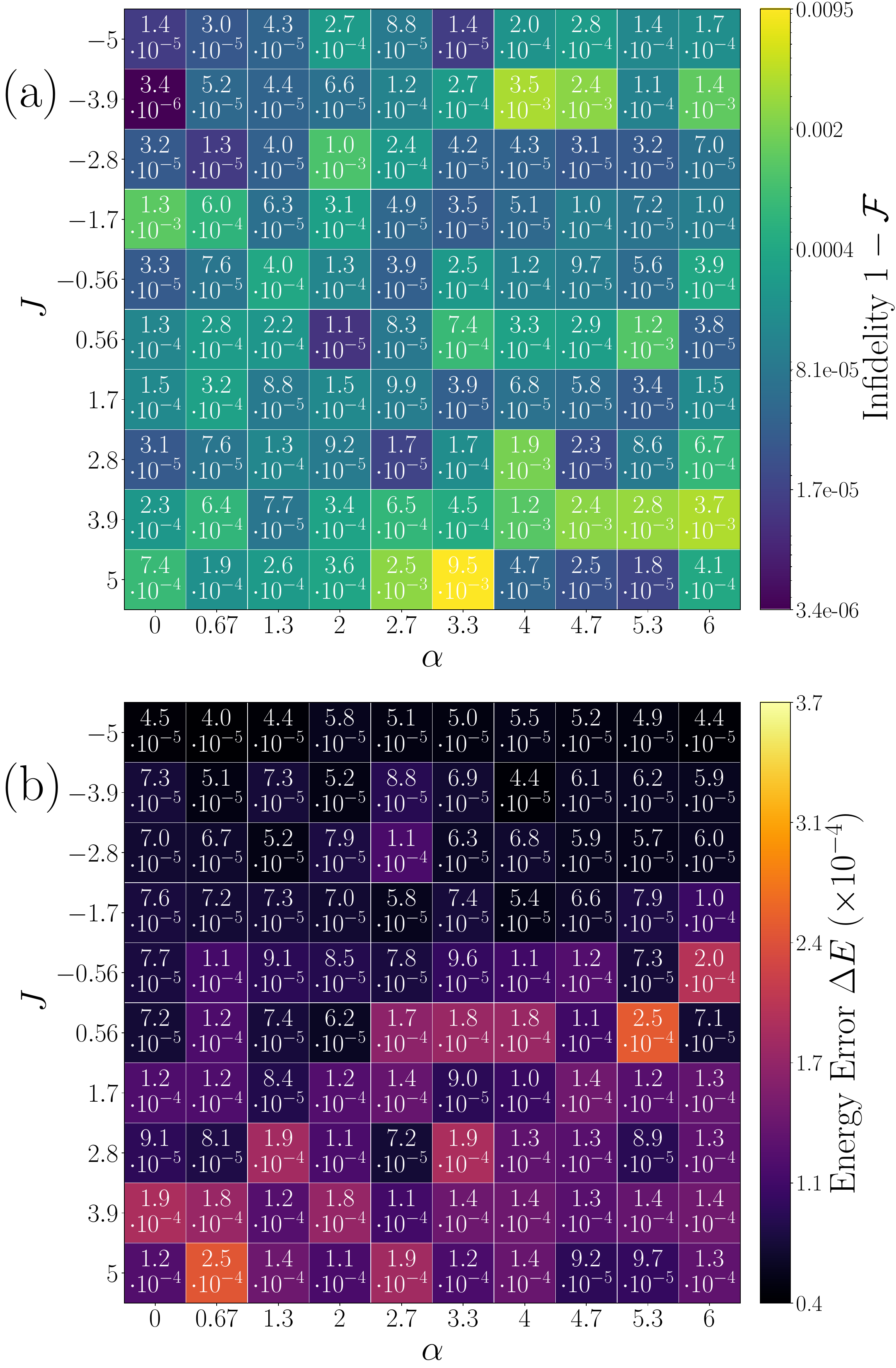}  
    \caption{\textbf{Benchmark of \cobalt\ on the long-ranged transverse-field Ising model at $N=20$.} We sweep the Hamiltonian parameters over a 2D grid—vertical axis $J\in[-5,5]$, horizontal axis $\alpha\in[0,6]$—and report two metrics per grid point. \textbf{Left panel:} the infidelity $1-F$, where $F$ is the overlap with the exact ground state; in degenerate regions we take the overlap with the full ground-state manifold. Values span from $\sim10^{-5}$ up to $\sim10^{-2}$, with most of the diagram at $\lesssim10^{-3}$ (lower is better). \textbf{Right panel:} the ground-state energy error $\Delta E=|E_{\text{NQS}}-E_0|$, reported at the $10^{-4}$ scale and typically in the $2\!-\!5\times10^{-4}$ range, with a few larger pockets. 
}
    \label{fig:infidelity}
\end{figure}

\section{Critical exponents and finite-size scaling}
To test whether DysonNet captures not only ground-state energies but also the long-distance structure of critical wavefunctions, we perform a finite-size-scaling analysis of the long-range TFIM at representative interaction exponents spanning the distinct universality regimes. Specifically, we compute the squared magnetization $\langle m^2\rangle$ for a range of system sizes and couplings near criticality and collapse the data using the scaling form
\begin{align}
N^{2\beta/\nu}\langle m^2\rangle = f\!\left((J-J_c)N^{1/\nu}\right),
\end{align}
with $J_c$, $\nu$, and $\beta$ treated as fit parameters. The collapse is obtained using the \texttt{PyFSSA} package~\cite{sorge_pyfssa_2015}, and the quoted uncertainties are estimated conservatively as the maximum of two procedures: a jackknife analysis in which one system size is omitted at a time, and a bootstrap resampling analysis over the full dataset.

The resulting collapses, shown in Fig.~\ref{fig:collapseTripple}, are of high quality across all three representative values of $\alpha$, demonstrating that DysonNet reproduces the universal critical behavior of the model rather than merely fitting local observables or variational energies. The extracted critical exponents listed in Table~\ref{tab:crit-exp-compare} agree closely with stochastic-series-expansion benchmarks and known theoretical predictions.  Together, these results show that the architecture remains accurate even at criticality and that access to larger system sizes substantially improves the reliability of exponent extraction.

\begin{figure*}
  \centering
  \subfloat[]{%
    \begin{overpic}[width=0.32\textwidth]{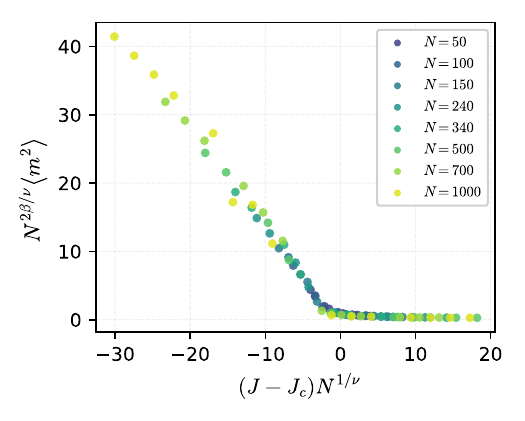}
      \put(5,71){\colorbox{white}{(a)}}
    \end{overpic}%
    \label{fig:collapse15}}
  \hfill
  \subfloat[]{%
    \begin{overpic}[width=0.32\textwidth]{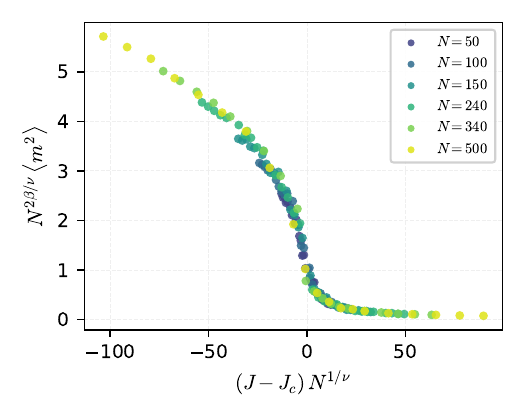}
      \put(5,71){\colorbox{white}{(b)}}
    \end{overpic}%
    \label{fig:collapse25}}
  \hfill
  \subfloat[]{%
    \begin{overpic}[width=0.32\textwidth]{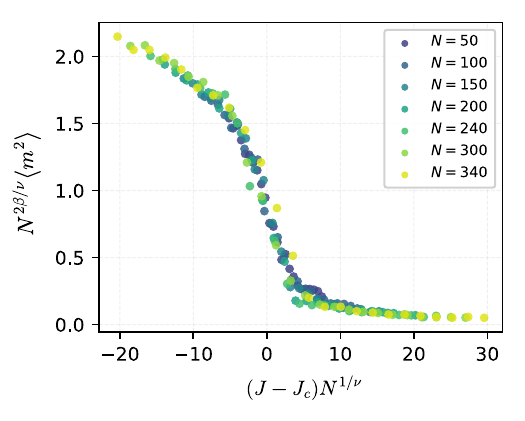}
      \put(5,71){\colorbox{white}{(c)}}
    \end{overpic}%
    \label{fig:collapse6}}
  \caption{ \textbf{Finite-size scaling analysis of the squared magnetization $\langle m^2 \rangle$ obtained via the \cobalt\ Ansatz.} The panels display data collapses for long-range decay exponents (a) $\alpha = 1.5$, (b) $\alpha = 2.5$, and (c) $\alpha = 6.0$. The data are rescaled according to the scaling ansatz $N^{2\beta/\nu} \langle m^2 \rangle$ versus $(J - J_c)N^{1/\nu}$, using the critical parameters listed in Table~\ref{tab:crit-exp-compare}. The high quality of the collapse across system sizes ranging from $N=50$ to $N=1000$ (for $\alpha=1.5$) validates the accuracy of the extracted critical exponents and critical points.}
  \label{fig:collapseTripple}
\end{figure*}

\section{Small system validation \texorpdfstring{$J_1$--$J_2$}{J1--J2}}

\label{app:j1j2}
The antiferromagnetic $J_1$--$J_2$ model is the spin-$1/2$ Heisenberg chain with competing nearest- and next-nearest-neighbor exchange,
\begin{equation}
    H = J_1 \sum_{i=1}^{N} \mathbf{\hat S}_i \!\cdot\! \mathbf{\hat S}_{i+1}
      + J_2 \sum_{i=1}^{N} \mathbf{\hat S}_i \!\cdot\! \mathbf{\hat S}_{i+2},
    \label{eq:j1j2_hamiltonian}
\end{equation}
with periodic boundary conditions. For $0<J_2/J_1\lesssim 0.24$ the ground state remains critical and is well described by a Luttinger liquid with algebraically decaying spin correlations, while increasing $J_2$ introduces frustration that reshapes short- and intermediate-range correlations~\cite{okamoto_fluid-dimer_1992,eggert_numerical_1996,white_dimerization_1996}. At $J_2/J_1\simeq 0.24$ the chain undergoes a Berezinskii--Kosterlitz--Thouless transition into a gapped, spontaneously dimerized phase, which continuously connects to the Majumdar--Ghosh point at $J_2/J_1=1/2$ where the ground state is exactly a product of singlet dimers (up to a twofold translation degeneracy)~\cite{okamoto_fluid-dimer_1992,majumdar_next-nearest-neighbor_1969}. This combination of (i) a well-established phase diagram containing both critical and gapped regimes, (ii) strong frustration and entanglement that stress variational expressivity, and (iii) exact-diagonalization access at small $N$ makes the $J_1$--$J_2$ chain a standard, stringent benchmark for many-body wavefunction ans\"atze~\cite{carleo_solving_2017,liang_solving_2018,viteritti_transformer_2023}.

Building on this standard benchmark, we evaluate the \cobalt\ architecture on the antiferromagnetic $J_1$--$J_2$ chain with $N=20$ sites across the frustration range $0 \le J_2/J_1 \le 0.5$ (Fig.~\ref{fig:j1j2_small_system}). \cobalt\ maintains high accuracy throughout, with the relative ground-state energy error $\varepsilon$ remaining below $4\times 10^{-4}$ for all couplings. The error reaches a minimum of $2.17\times 10^{-5}$ at $J_2/J_1 \approx 0.08$ and then increases smoothly with frustration, attaining a maximum of $3.71\times 10^{-4}$ near the strongly frustrated regime. Importantly, we observe no sharp loss of precision in the vicinity of the quantum critical point at $J_2/J_1 \approx 0.24$~\cite{okamoto_fluid-dimer_1992,eggert_numerical_1996}, suggesting that \cobalt\ captures the ground-state correlations consistently across the transition. The variational V-score mirrors this trend, ranging from $1.91\times 10^{-4}$ to $2.44\times 10^{-3}$.

\begin{figure}[H]
    \centering
    \includegraphics[width=\linewidth]{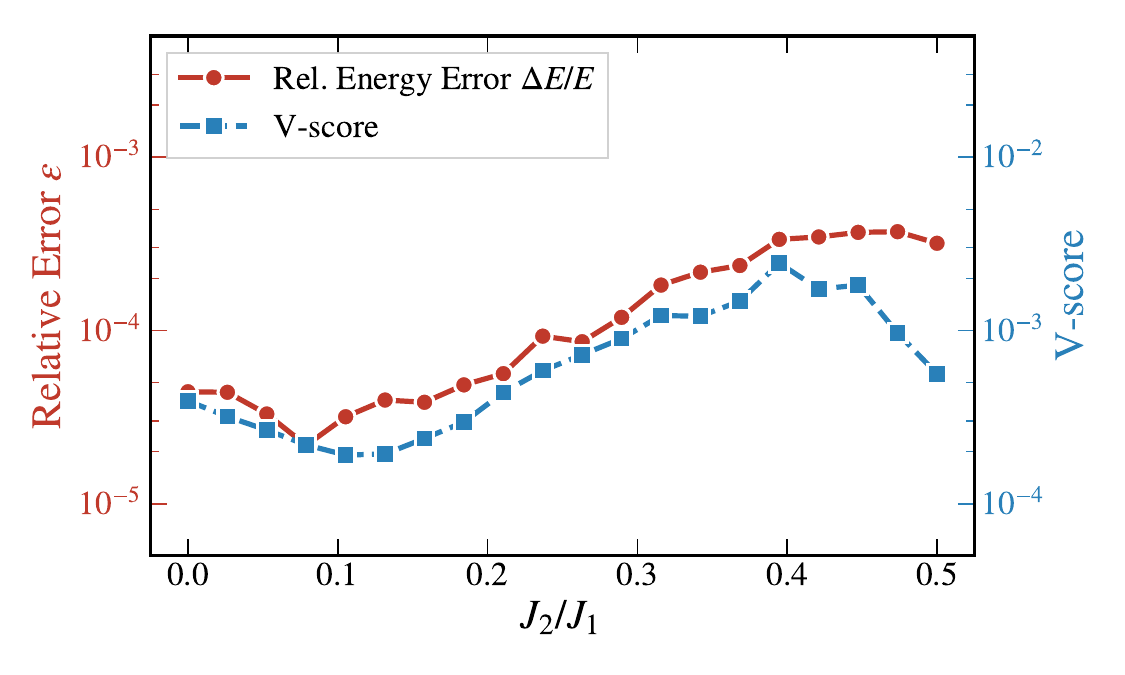}
    \caption{\textbf{Benchmarks of \cobalt\  applied to the 1D antiferromagnetic $J_1$--$J_2$ spin chain with $N=20$ spins.} The plot compares the relative error in the ground-state energy, defined as $\varepsilon = |E_{\text{\cobalt}} - E_{\text{ED}}|/|E_{\text{ED}}|$ (red circles, left logarithmic axis), and the variational V-score (blue squares, right logarithmic axis) as a function of the frustration ratio $J_2/J_1$. Reference ground-state energies $E_{\text{ED}}$ are obtained via exact diagonalization. The results illustrate the ability of \cobalt\  to maintain high accuracy across the frustration regime, with the V-score serving as a reliable unsupervised proxy for the true energy error.}
    \label{fig:j1j2_small_system}
\end{figure}

\section{Network parameters}
\label{app:training_details}
Unless otherwise stated, optimization uses standard SR and TF32 precision on NVIDIA GPUs. Sweep counts are scaled linearly with system size $N$ in all phases, including area-law regimes. Chains are randomly initialized, the burn-in length is five sweeps, and global spin inversion is attempted after every $\left\lceil N/(2w+1)\right\rceil$ proposals. For the RBM baseline we use hidden-unit density $\alpha_{\mathrm{RBM}}=10$ (\(N_h=10N_v\)), which is distinct from the Hamiltonian decay exponent $\alpha$. Most training runs use NVIDIA Tesla P100, T4, and L4 GPUs; the largest runs on criticality ($N>600$) are executed on a single NVIDIA A40. We note or $N<500$ DysonNet can be trained in google collab or Kaggle notebooks, meaning DysonNet makes large system sizes generally accessible without the need for high-end  compute clusters. For finite-size scaling we warm-start larger systems from converged smaller systems. Critical exponents are extracted with PyFSSA~\cite{sorge_pyfssa_2015}. Uncertainties for the critical exponents are computed by taking the maximum over the bootstrapping and jacknife approach.  
\onecolumngrid

\begin{table}[htbp]
\centering
\caption{Matched architecture hyperparameters for Fig.~\ref{fig:kernel_runtime}(b,c). ViT and linearized ViT are configured identically.}
\label{tab:runtime_matched_hparams}
\begin{tabular}{lcc}
\toprule
\textbf{Hyperparameter} & \textbf{ViT / Linearized ViT} & \textbf{\cobalt} \\
\midrule
Layers / Blocks & 2 & 2 \\
Token size & 2 & 2 \\
Embedding dimension & 2 & 2 \\
FFN layers per block & 1 & 1 \\
LogCosh hidden units & 2 & 2 \\
Attention heads & 2 & N/A \\
S4 state-space dimension & N/A & 12 \\
\bottomrule
\end{tabular}
\end{table}

\begin{table}[htbp]
    \centering
    \caption{Hyperparameters for Network, Sampler, and Optimizer (Run 1).}
    \label{tab:hyperparameters_run1}
    \begin{tabular}{l l}
        \toprule
        \textbf{Parameter} & \textbf{Value / Scaling} \\
        \midrule
        \multicolumn{2}{c}{\textit{Network Architecture (\cobalt)}} \\
        \midrule
        Block Type & \cobalt (S4 kernels, bidirectional complex) \\
        Number of Layers ($N_{layers}$) & 2  \\
        Embedding Dimension ($d_{model}$) & 14 \\
        S4 State Dimension ($N_{s4}$) & 12  \\
        Token Size & 2 \\
        LogCosh Hidden Units & 2 \\
        Convolution Kernel Size & 4 \\
        S4 initialization & Hippo \\ 
        Normalization & ActNorm \\
        \midrule
        \multicolumn{2}{c}{\textit{Sampler (Metropolis)}} \\
        \midrule
        Number of Chains ($N_{chains}$) & 512 \\
        Number of Samples ($N_{samples}$) & 2048 \\
        Chain Initialization & Random spin configurations \\
        Burn-in Length & 5 sweeps \\
        Global Spin Inversion Interval & $\left\lceil N/(2w+1)\right\rceil$ proposals \\
        ABACUS & Enabled if $N > 50$ \\
        Fast Sampler Strategy & Enabled if $N > 100$ \\
        Sweep Size ($S_{sweep}$) & $\approx 1.4 \times S_{base}$ \\
        \quad \textit{Base Sweep Unit} ($S_{base}$) & $N \cdot C_{crit}$ \\
        \quad \textit{Scaling with System Size} & Linear in $N$ (all phases) \\
        \quad \textit{Critical Factor} ($C_{crit}$) & $4 \cdot \max(0.8 - |J - J_c|, 0) + 1$ \\
        \midrule
        \multicolumn{2}{c}{\textit{Optimizer (SR/SGD with Schedule)}} \\
        \midrule
        SR Variant & Standard SR (no minSR) \\
        Training Iterations & 400 \\
        Warmup Steps & 150 \\
        Peak Learning Rate ($\eta$) & $\max(3.5, 5.0 \cdot \frac{N}{340})$ \\
        Diagonal Shift Schedule & Linear ($10^{-2} \to 10^{-4}$) \\
        \bottomrule
    \end{tabular}
\end{table}

\begin{table}[htbp]
    \centering
    \caption{Transformer (SpinViT) Architecture Hyperparameters (Run 1).}
    \label{tab:hyperparameters_spinvits_run1}
    \begin{tabular}{l l}
        \toprule
        \textbf{Parameter} & \textbf{Value / Description} \\
        \midrule
        \multicolumn{2}{c}{\textit{Network Architecture (BatchedSpinViT)}} \\
        \midrule
        Model Type & BatchedSpinViT (Transformer / ViT-style) \\
        Linearized ViT Baseline & Identical hyperparameters \\
        Token Size & 5 \\
        Embedding Dimension ($d_{\text{model}}$) & 14 \\
        Number of Attention Heads ($N_{\text{heads}}$) & 2 \\
        Number of Transformer Blocks ($N_{\text{blocks}}$) & 2 \\
        FFN Layers per Block ($N_{\text{ffn}}$) & 3 \\
        Final Architecture & $(5,)$ \\
        Complex-Valued Model & False \\
        \bottomrule
    \end{tabular}
\end{table}

\end{document}